\documentclass[lettersize,journal]{IEEEtran}
\IEEEoverridecommandlockouts

\usepackage{amssymb}
\usepackage[cmex10]{amsmath}
\usepackage{stfloats}
\usepackage{graphicx}
\usepackage{subfigure}
\usepackage{tabularx}
\usepackage{epsfig,epsf,color,balance,cite}
\usepackage{verbatim}
\usepackage{url}
\usepackage{bm}
\usepackage{booktabs}
\usepackage{siunitx}

\usepackage{amsmath}
\allowdisplaybreaks[2]
\usepackage{ntheorem}
\newtheorem{theorem}{Theorem}
\newtheorem{lemma}{Lemma}
\newtheorem*{proof}{Proof}

\usepackage{algorithm}
\usepackage{algorithmic}

\usepackage[caption=false,font=footnotesize]{subfig}

\usepackage{color}
\definecolor{myc1}{rgb}{0,0,0}

\begin{document}

\title{Digital Twin Synchronization Over Mobile Embodied AI Network With Agentic Intelligence}

\author{Zhouxiang Zhao,
        Jiaxiang Wang,
        Yahao Ding,
        Yinchao Yang,
        Zhaohui Yang,
        Mohammad Shikh-Bahaei,\\
        Julie A. McCann,
        Zhaoyang Zhang,~\IEEEmembership{Senior Member,~IEEE,}
        and Kaibin Huang,~\IEEEmembership{Fellow,~IEEE}
        
\thanks{Zhouxiang Zhao, Zhaohui Yang, and Zhaoyang Zhang are with the College of Information Science and Electronic Engineering, Zhejiang University, Hangzhou 310027, China (e-mails: \{zhouxiangzhao, yang\_zhaohui, ning\_ming\}@zju.edu.cn).}
\thanks{Jiaxiang Wang, Yahao Ding, and Mohammad Shikh-Bahaei are with the Department of Engineering, King's College London, London, WC2R 2LS, UK (e-mails: \{jiaxiang.wang, yahao.ding, m.sbahaei\}@kcl.ac.uk).}
\thanks{Yinchao Yang and Julie A. McCann are with the Department of Computing, Imperial College London, SW7 2AZ London, U.K. (e-mails: \{yinchao.yang, j.mccann\}@imperial.ac.uk).}
\thanks{Kaibin Huang is with the Department of Electrical and Electronic Engineering, The University of Hong Kong, Hong Kong SAR, China (e-mail: huangkb@eee.hku.hk).}
}

\maketitle

\begin{abstract}

Efficient digital twin (DT) synchronization relies on maintaining high-fidelity virtual representations with minimal age of information (AoI). However, the synergistic potential of cooperative sensing and autonomous mobility of the sensing agent remains underexplored in existing DT synchronization frameworks. In this paper, we propose an agentic AI-empowered mobile embodied AI network (MEAN) framework for DT synchronization. In the proposed hybrid architecture, the base station (BS) conducts global orchestration, while the agents autonomously execute a five-stage closed-loop workflow: move-to-sense, cooperative sensing, onboard semantic processing, channel-aware mobility, and uplink transmission. To optimize synchronization performance, we formulate a joint topology dispatching and multidimensional resource allocation problem aimed at minimizing the maximum twin deviation across regions, subject to heterogeneous sensing fidelity and energy budget constraints. To tackle this, we develop a hierarchical two-layer optimization algorithm, where the outer-layer refines multi-agent assignment via a dynamic matching game, and the inner-layer iteratively optimizes the continuous resources. Extensive simulation results verify the convergence of the proposed algorithm and demonstrate its substantial superiority over multiple baseline schemes in reducing synchronization deviation. Furthermore, the results reveal that semantic compression serves as a vital substitute for channel resources in latency reduction under constrained bandwidth, while autonomous velocity adaptation provides an essential degree of freedom for the system to navigate the fundamental energy-time trade-off.
\end{abstract}


\IEEEpeerreviewmaketitle

\section{Introduction}
\IEEEPARstart{T}{he} rapid evolution of digital twin (DT) is transforming how complex physical systems are monitored, analyzed, and controlled across a wide range of domains, including smart manufacturing, healthcare, security, and wireless networks \cite{9899718}. Unlike conventional simulation tools that rely on static models and predefined assumptions, a DT operates as a living model that continuously integrates real-time data from the physical world \cite{11207697}. This tight coupling allows DTs to accurately reflect system states, anticipate future behaviors, and support adaptive optimization, thereby unlocking significant improvements in efficiency, reliability, and autonomy.

Recent advances in sensing technologies, particularly cooperative and distributed sensing through mobile embodied AI networks (MEANs), have substantially enhanced the fidelity of DTs by providing rich, multi-source observations across spatially distributed environments \cite{chen2026split,wu2026joint,9921194}. In parallel, modern communication technologies enable MEANs to transmit large volumes of data to edge or cloud platforms with low latency and high reliability \cite{ZHAO2024107055,10841377,11006980}. Building upon this sensing and communication foundation, embodied AI and agentic intelligence further enhance DT construction and maintenance by enabling MEANs to autonomously adapt their sensing, communication, data processing, and decision-making strategies, thereby enhancing overall DT fidelity and scalability. 

\subsection{Related Works}

\subsubsection{Digital Twin Synchronization} 

DT synchronization involves the timely alignment of the virtual representation with the current state of the physical source \cite{zhang2025transferring}. 
Maintaining tight synchronization is critical, as outdated or inconsistent state information can degrade DT fidelity. Several studies have investigated DT synchronization from a resource allocation perspective, including selection of access networks \cite{10070572}, selection of protocols \cite{cakir2023synchronize}, uplink transmission scheduling \cite{yu2025optimizing}, spectrum allocation \cite{elloumi2025spectrum}, and optimization of computational resources and energy \cite{liu2024two}.

Complementary to resource allocation approaches, other works focus on reducing the data volume required for DT updates through semantic compression techniques. For example, \cite{10530992} and \cite{li2025delay} exploit semantic communication to encode sensing data into compact representations that preserve essential meaning and features. By transmitting semantically relevant information instead of raw data, these methods significantly reduce uplink latency, thereby facilitating more efficient and timely DT synchronization.

\subsubsection{Age of Information} 

Age of information (AoI) quantifies the timeliness of an orchestrator’s knowledge of a remote entity or process, by measuring the elapsed time since the most recently received update \cite{9380899}. In DT systems, maintaining low AoI, together with high-quality sensing, is critical for ensuring accurate and up-to-date state awareness of the physical system. A fundamental trade-off arises whereby longer sampling and observation intervals improve data quality and estimation accuracy but inevitably increase AoI.

Recent works have investigated AoI optimization in DT systems to enable high-fidelity state awareness. In particular, various approaches have been proposed to manage bandwidth \cite{10000964, liao2024integration}, energy \cite{guo2024age, liu2025joint}, and computational resources \cite{11183623, 11004606}, to constrain AoI within stringent thresholds. Furthermore, \cite{10460140} examines the intricate interplay between sensing quality, sensing cost, and sensing redundancy, highlighting trade-offs between temporal freshness and data acquisition efficiency.

\subsubsection{Embodied Agentic AI}

Embodied agentic AI integrates intelligence into physical platforms, enabling systems to perceive, interact with, and navigate their environments by harnessing sensing, communication, computation, and actuation capabilities. Such agents can autonomously pursue complex objectives with minimal human intervention \cite{zhao2025agentic,11370176, 11339915}.

Existing studies have investigated several aspects of agentic AI, including its integration with edge computing infrastructures \cite{11373363,11366905}, the design of agent-to-agent (A2A) interaction and coordination mechanisms \cite{11303308,11373008}, and the development of multi-modal perception techniques \cite{11297177,11298134}. In parallel, prior work has examined a broad range of application scenarios enabled by agentic AI \cite{11207716,zhao2026agentic}, as well as the incorporation of advanced large language models (LLMs) to support reasoning and decision-making \cite{11303197,11022699}.

\subsection{Contributions}
Despite the extensive efforts in the aforementioned areas, critical research gaps remain. First, existing DT synchronization schemes primarily focus on static resource allocation or data compression, largely overlooking the synergistic potential of multi-agent cooperative sensing and autonomous mobility. Second, current AoI-oriented studies seldom consider how embodied agents can actively reshape channel conditions through spatial exploration to mitigate transmission bottlenecks. Third, while agentic AI has demonstrated remarkable potential in high-level reasoning, there lacks a quantitative optimization framework that harnesses its capabilities for physical-virtual synchronization in dynamic communication networks.

To address these research gaps, our work integrates cooperative sensing, onboard semantic processing, and autonomous channel-aware mobility to explicitly navigate the energy-time trade-off in DT systems. 
Our main contributions are summarized as follows:
\begin{itemize}
    \item We propose a comprehensive MEAN framework tailored for DT synchronization, which is characterized by a five-stage closed-loop workflow: move-to-sense, cooperative sensing, onboard semantic processing, channel-aware mobility, and uplink transmission. In this hybrid framework, the base station (BS) acts as a central orchestrator, while the distributed agents autonomously adapt their agentic intelligence in the physical execution layer to dynamically evolving environments.
    \item We formulate a cooperative sensing model that quantitatively characterizes the non-linear relationship between collaborative sensing duration and achieved perception accuracy, which is empirically validated utilizing real-world datasets. Furthermore, we establish a channel-aware autonomous mobility model that mathematically captures the statistical spatial diversity gain achieved through active environmental exploration.
    \item To holistically evaluate the synchronization performance, we formulate a mixed-integer non-linear programming (MINLP) problem aimed at minimizing the maximum twin deviation, subject to stringent sensing fidelity and energy budget constraints. To achieve efficient resource scheduling and optimization, we develop a hierarchical two-layer algorithm. The outer-layer refines the multi-agent assignment via a dynamic matching game with heuristic guidance, while the inner-layer optimizes the continuous resources iteratively.
\end{itemize}

The remainder of this paper is organized as follows. Section \ref{Sec:smpf} establishes the MEAN-enabled DT synchronization framework and formulates the twin deviation minimization problem. Section \ref{Sec:ad} develops a hierarchical two-layer optimization algorithm. Section \ref{Sec:sra} presents comprehensive numerical results and performance analysis to validate the proposed scheme. Finally, Section \ref{Sec:c} concludes the paper.


\section{System Model and Problem Formulation}\label{Sec:smpf}
In this section, we formulate the mathematical model for the proposed DT-native MEAN.

\begin{figure}[t]
    \centering
    \includegraphics[width=\linewidth]{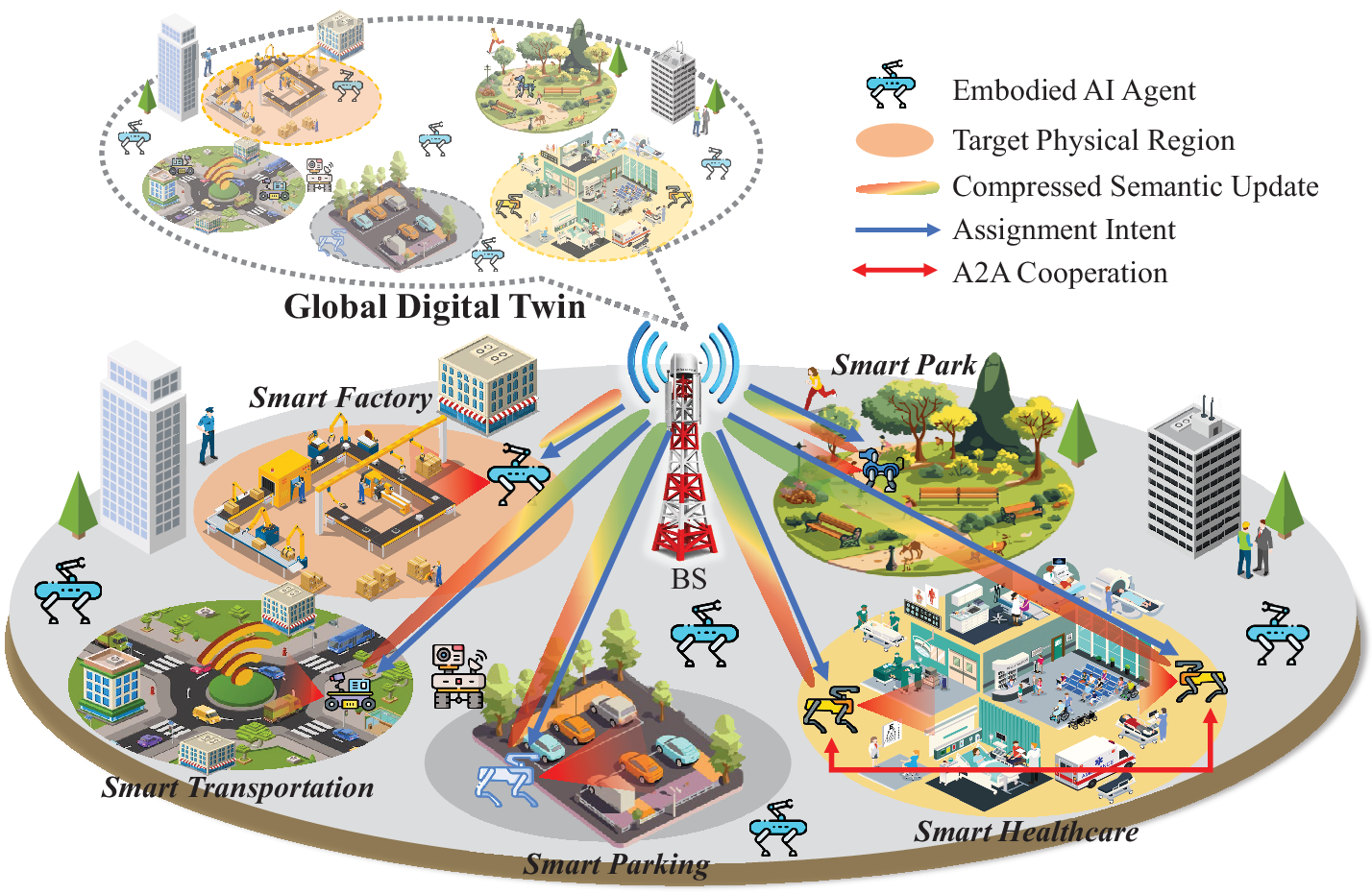}
    \caption{An illustration of the considered DT synchronization scenario over agentic AI-empowered MEAN.}
    \label{fig.scenario}
\end{figure}

\subsection{Scenario Overview}
Consider an agentic AI-empowered MEAN deployed in a dynamic physical environment, comprising a BS and a set of ground agents, denoted by $\mathcal{N} = \{1, 2, \dots, N\}$. 
The BS serves as a central orchestrator that maintains a global DT, as illustrated in Fig.~\ref{fig.scenario}. 
The system operates under a centralized orchestration and decentralized execution paradigm. The BS acts as the central brain, while the agents exhibit decentralized agentic autonomy in the physical layer.


To accurately capture the real-world dynamics, the DT must be synchronized with the physical entities. However, due to the inherent physical volatility, the virtual states in the DT inevitably deviate from the physical reality over time. Once the DT identifies that the twin deviation of certain critical areas exceeds a tolerable threshold, an on-demand update request is triggered. Let $\mathcal{K} = \{1, 2, \dots, K\}$ denote the set of physical regions targeted for the current DT update. In a three-dimensional (3D) Cartesian coordinate system, the BS is elevated at a fixed location $\mathbf{x}_{\text{BS}} = [0, 0, H]^\text{T}$. Each target region $k \in \mathcal{K}$ is geometrically modeled as a circular coverage area on the horizontal plane, characterized by its center $\mathbf{c}_k = [x_k, y_k, 0]^\text{T}$ and radius $r_k$.

When an update is triggered, the agents, which are initially dispersed at coordinates $\mathbf{x}_n^{\text{init}} = [x_n^{\text{init}}, y_n^{\text{init}}, 0]^\text{T}$ performing their routine tasks, are dynamically drafted by the BS. To coordinate this multi-agent dispatching, we define a binary assignment indicator $\alpha_{n,k} \in \{0, 1\}$, where $\alpha_{n,k} = 1$ indicates that agent~$n$ is dispatched to update region $k$, and $\alpha_{n,k} = 0$ otherwise.


\subsection{Move-to-Sense}
Upon receiving the dispatch intent from the DT orchestrator, the assigned agent $n$ (i.e., $\alpha_{n,k} = 1$) is required to autonomously navigate from its initial coordinate $\mathbf{x}_n^{\text{init}}$ to the effective sensing coverage of the target region $k$.

Without loss of generality, we assume the agent initiates the sensing task when it enters the specific region. Depending on the initial spatial distribution, we uniformly formulate the actual mobility distance $d_{n,k}^\text{ms}$ required for agent $n$ to reach the sensing region~$k$ as follows:
\begin{equation}
    d_{n,k}^\text{ms} = \max \left\{ 0, \lVert \mathbf{x}_n^{\text{init}} - \mathbf{c}_k \rVert_2 - r_k \right\},
\end{equation}
which handles the scenario where the agent is already located within the target region, yielding $d_{n,k}^\text{ms} = 0$.

Let $v_n^{\text{ms}} \in (0, v_n^{\max}]$ denote the constant cruising speed of agent $n$ traveling towards region $k$, where $v_n^{\max}$ is the maximum mechanical speed limit. The mobility delay incurred in this phase is given by:
\begin{equation}
    t_{n,k}^\text{ms} = \frac{d_{n,k}^\text{ms}}{v_n^{\text{ms}}}.
\end{equation}

For ground mobile agents operating at relatively low speeds, the air resistance is negligible. The dominant mechanical power consumption stems from overcoming the friction of the ground and the internal heat loss of the motor \cite{1638342}. Consequently, we model the mechanical mobility power $P_n^{\text{ms}}$ as a polynomial function comprising the first-order and second-order terms of the velocity:
\begin{equation}
    P_n^{\text{ms}} = \lambda_1 v_n^{\text{ms}} + \lambda_2 (v_n^{\text{ms}})^2,
\end{equation}
where $\lambda_1 > 0$ and $\lambda_2 > 0$ are the hardware-specific kinetic coefficients related to the rolling resistance and the motor's internal impedance, respectively.

Based on the above power formulation, the total energy consumption for agent $n$ during the move-to-sense phase can be derived as:
\begin{equation}
    E_{n,k}^{\text{ms}} = P_n^{\text{ms}} t_{n,k}^\text{ms} = d_{n,k}^\text{ms} \left(\lambda_1 + \lambda_2 v_n^{\text{ms}}\right).
\end{equation}

\subsection{Cooperative Sensing for DT}
Once agent $n$ enters its assigned region $k$, it activates its onboard sensors to capture the physical state. In our DT framework, the virtual reconstruction of a physical region benefits from multi-agent collaboration. Let $N_k = \sum_{n=1}^N \alpha_{n,k}$ denote the total number of collaborative agents dispatched to region $k$. Notably, when $N_k > 1$, the agents do not redundantly sense the same target. Instead, they leverage decentralized A2A communication to dynamically allocate sensing tasks and coordinate their field of view. By intelligently assigning distinct physical objects or sub-areas within region $k$ to different agents, the multi-agent cluster achieves a comprehensive and distributed coverage. Subsequently, as these agents perform cooperative scanning from different perspectives, the DT orchestrator can exploit the spatial view orthogonality to patch the holistic 3D semantic model, e.g., resolving depth ambiguity, clearing occlusions, and fusing textures.

\begin{figure}[t]
    \centering
    \includegraphics[width=\linewidth]{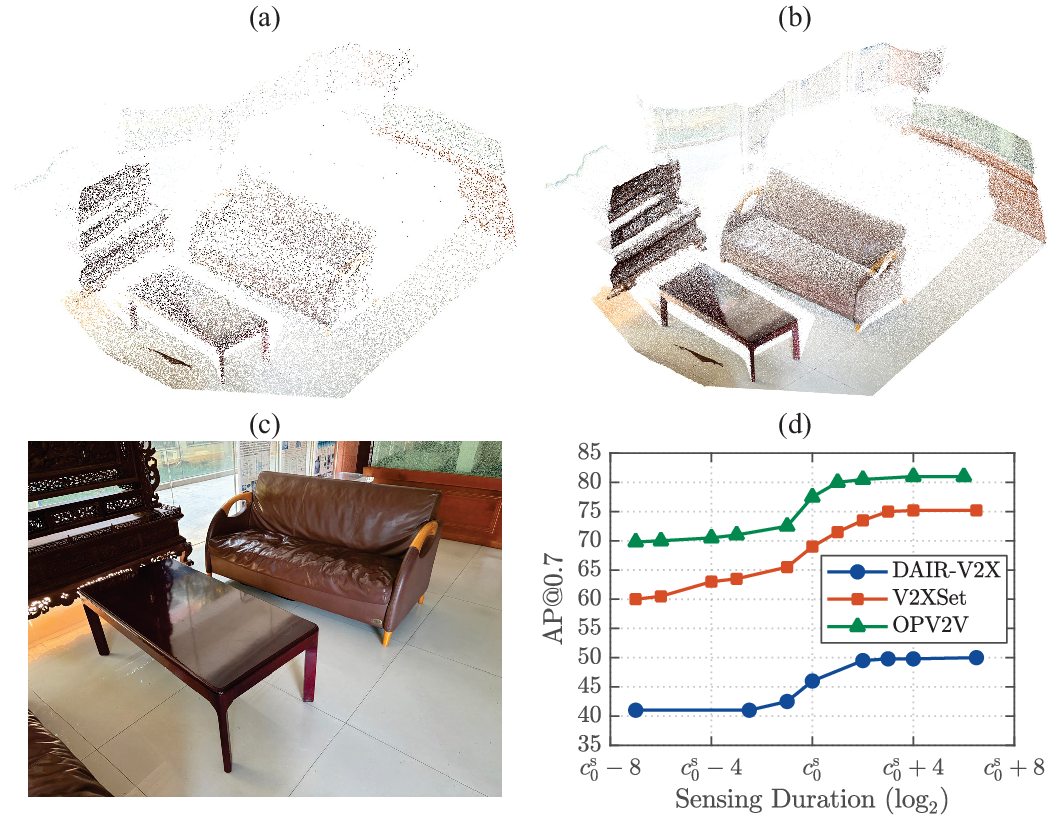}
    \caption{A visualization of sensing accuracy: (a) 3D reconstruction results based on short-duration sensing data, (b) 3D reconstruction results based on long-duration sensing data, (c) region setup, (d) average precision versus sensing duration across different datasets, including DAIR-V2X \cite{9879243}, V2XSet \cite{xu2022v2x}, and OPV2V \cite{9812038}.}
    \label{fig.sensing}
\end{figure}

To quantitatively characterize the relationship between the sensing duration and the achieved accuracy, we conduct preliminary perception simulation experiments. As illustrated in Fig.~\ref{fig.sensing}, the accuracy curves evaluated across diverse datasets consistently exhibit a characteristic sigmoid growth trend~\cite{yang2023spatio}. Furthermore, to capture the collaborative gain derived from the aforementioned A2A task allocation, we introduce a factor $\theta \in (0,1)$ and define the cooperative scaling function as $g(N_k) = N_k^\theta$ \cite{9812038}. By integrating this cooperative multiplier into the fitted sigmoid model, the equivalent sensing accuracy $q_{n,k}$ achieved by agent $n$ over a continuous sensing duration $t_n^{\text{s}}$ can be analytically formulated as:
\begin{equation}\label{eq.cs}
    q_{n,k} = \frac{\sigma}{1 + 2^{-\xi \left[\log_2\left(t_n^{\text{s}} N_k^\theta\right) - c_0^\text{s}\right]}},
\end{equation}
where $\sigma\in (0,1)$ is the model performance upper bound, $\xi > 0$ is the curve-fitted steepness parameter, and $c_0^\text{s}$ denotes the threshold where the sensing quality experiences a qualitative leap. Equation \eqref{eq.cs} explicitly reveals that dispatching more agents to a single region effectively reduces the individual sensing time required to achieve the same target accuracy.

Assuming a constant effective sampling rate $\gamma_n$ for the onboard sensors, the raw data volume $D_n$ acquired by agent $n$ is proportional to its sensing time:
\begin{equation}
    D_n = \gamma_n t_n^{\text{s}}.
\end{equation}

The energy consumed in the cooperative sensing phase can be given by:
\begin{equation}
    E_n^{\text{s}} = P_n^{\text{s}} t_n^{\text{s}},
\end{equation}
where $P_n^{\text{s}}$ represents the power consumed during the sensing phase.

\subsection{Compute while Moving}
Upon accomplishing the sensing task, the agent enters a coupled integrated computation and mobility phase. Specifically, it leverages its onboard computing resources to extract semantic features and compress the raw data. Simultaneously, the agent autonomously navigates its environment to identify a location with superior uplink channel conditions. This searching mobility is executed in parallel with the local computing to minimize the overall latency.

The considered DT-native framework can exploit the historical and spatial correlations inherent in the twin model. Let $Q_k \in [0,1)$ denote the prior prediction confidence of the global DT regarding the current physical state of region $k$. A higher $Q_k$ indicates that the DT possesses an accurate prior estimation, and thus the agent only needs to transmit the semantic residuals or state deviations.

To mathematically formalize this context-aware interaction, we define $\rho_{n,k} \in (0,1]$ as the semantic compression ratio of agent $n$ \cite{10550151}. The lower bound of $\rho_{n,k}$ is governed by the DT's prior confidence as follows:
\begin{equation}\label{eq.rholb}
    \rho_{n,k} \ge \rho_n^{\min} (1 - Q_k),
\end{equation}
where $\rho_n^{\min} \in (0,1)$ characterizes the fundamental compression limit of the local foundation model without any prior knowledge. Equation \eqref{eq.rholb} indicates that as the twin uncertainty $(1 - Q_k)$ decreases, the agent is permitted to adopt a more aggressive compression strategy, thereby alleviating the transmission burden.

The computation complexity for semantic extraction is modeled as $C(\rho_{n,k}) = \eta_n \ln(1/\rho_{n,k})$ operations per bit, where $\eta_n$ is a scaling parameter \cite{11208653,10915662}. Given the raw data volume $D_n$ acquired in the previous phase, the computation delay $t_{n,k}^{\text{c}}$ is formulated as:
\begin{equation}\label{eq.tc}
    t_{n,k}^{\text{c}} = \frac{D_n C(\rho_{n,k})}{f_n} = \frac{\gamma_n t_n^{\text{s}} \eta_n \ln(1/\rho_{n,k})}{f_n},
\end{equation}
where $f_n$ denotes the allocated computation capacity of agent~$n$.

Accordingly, the corresponding energy consumption for local computing is given by:
\begin{equation}\label{eq.ec}
    E_{n,k}^{\text{c}} = \kappa f_n^3 t_{n,k}^{\text{c}} = \kappa f_n^2 \gamma_n t_n^{\text{s}} \eta_n \ln\frac{1}{\rho_{n,k}},
\end{equation}
where $\kappa$ is the effective switched capacitance of the local processor.

Concurrently, the agent intelligently explores the physical environment to establish a reliable communication link. Let $t_n^{\text{mt}}$ and $v_n^{\text{mt}} \in (0, v_n^{\max}]$ denote the duration and cruising speed of this autonomous mobility phase, respectively. The energy consumed during this channel searching phase is:
\begin{equation}
    E_n^{\text{mt}} = P_n^\text{mt} t_n^{\text{mt}} = v_n^{\text{mt}} t_n^{\text{mt}} \left(\lambda_1 + \lambda_2 v_n^{\text{mt}}\right).
\end{equation}

Since the local computing and the autonomous mobility are executed in parallel, the effective duration $t_{n,k}^{\text{p}}$ of this integrated phase is determined by the longer of the two processes:
\begin{equation}\label{eq.tp}
    t_{n,k}^{\text{p}} = \max \left\{ t_{n,k}^{\text{c}}, t_n^{\text{mt}} \right\}.
\end{equation}
By the end of this parallel phase, the agent discovers an optimized transmission spot through its spatial exploration over the moving distance $d_n^{\text{mt}} = v_n^{\text{mt}} t_n^{\text{mt}}$, which will statistically enhance the subsequent transmission.

\subsection{Transmission Model}
Assume that the agent remains static while transmitting data. This stationary posture avoids the Doppler shifts and fast channel fading caused by mobility, thereby ensuring a highly reliable and stable uplink connection for uploading the compressed semantic payload to the DT orchestrator.

Due to the highly localized volatility of small-scale multipath fading, instantaneous channel conditions fluctuate significantly across space. However, environmental exploration improves the exploitable link quality by allowing an agent to proactively probe nearby positions and transmit from the most favorable one, effectively executing a search-and-select mechanism~\cite{Bonilla2016MDA}. Meanwhile, path loss-based signal maps and channel knowledge maps provide a natural reference for the location-dependent channel conditions of a region~\cite{Miyagusuku2018GPPathLoss,11414114}. Statistically, a larger exploration distance yields a larger spatial sample size, thereby increasing the probability of discovering a location with superior channel conditions. To mathematically capture this autonomous channel-searching capability, let $h_k^{\text{ref}} = \beta_0 (\lVert \mathbf{c}_k \rVert_2^2 + H^2)^{-\delta/2}$ denote the reference channel power gain from region $k$ to the BS, where $\beta_0$ represents the channel power gain at the reference distance of \SI{1}{m}, and $\delta$ is the path loss exponent. Driven by this search-and-select process, the expected maximum effective channel gain $h_{n,k}$ achieved by agent $n$ in region $k$ can be modeled as:
\begin{equation}
    h_{n,k} = h_k^{\text{ref}} + \phi_k(d_n^{\text{mt}}),
\end{equation}
where $\phi(\cdot) \ge 0$ represents the statistical spatial diversity gain expected from the active exploration. Without loss of generality, we model $\phi_k(x) = \omega_k \left(1 - e^{-x / \zeta_k}\right)$, where $\omega_k \ge 0$ and $\zeta_k > 0$ are environment-specific coefficients of region $k$.

\begin{figure}[t]
    \centering
    \includegraphics[width=\linewidth]{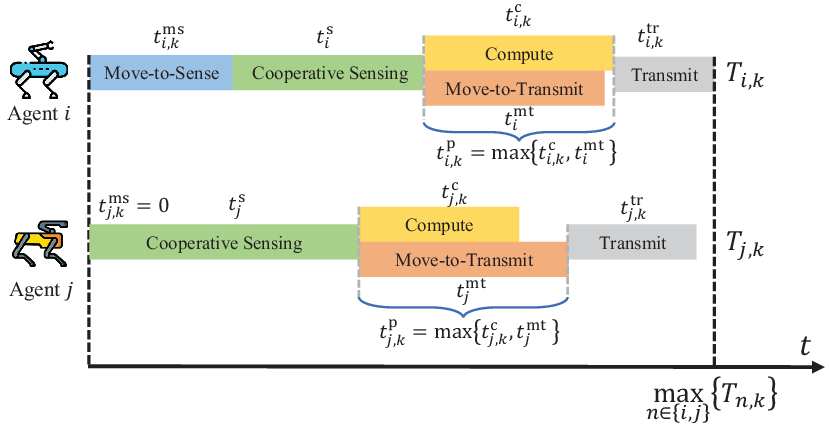}
    \caption{Closed-loop workflow of two agents assigned to the same region.}
    \label{fig.workflow}
\end{figure}

To mitigate co-channel interference among the spatially distributed agents, we assume that the system adopts orthogonal frequency-division multiple access (OFDMA). Let $b_n$ and $p_n$ denote the allocated bandwidth and the transmit power of agent $n$, respectively. Then, the achievable uplink transmission rate $R_n$ can be given as:
\begin{equation}
    R_{n,k} = b_n \log_2 \left( 1 + \frac{p_n h_{n,k}}{N_0 b_n} \right),
\end{equation}
where $N_0$ denotes the power spectral density of the additive white Gaussian noise (AWGN) at the BS receiver.

Consequently, the transmission delay $t_{n,k}^{\text{tr}}$ can be given by:
\begin{equation}
    t_{n,k}^{\text{tr}} = \frac{\rho_{n,k} D_n}{R_{n,k}} = \frac{\rho_{n,k} \gamma_n t_n^{\text{s}}}{b_n \log_2 \left( 1 + \frac{p_n h_{n,k}}{N_0 b_n} \right)}.
\end{equation}
The corresponding energy consumed for this uplink transmission is $E_{n,k}^{\text{tr}} = p_n t_{n,k}^{\text{tr}}$.

Fig.~\ref{fig.workflow} depicts the complete update workflow of the agent to synchronize the DT.

\subsection{Digital Twin Update Quality}
In a multi-agent cooperative DT network, the system must simultaneously guarantee the spatial fidelity of the sensed data and minimize the temporal staleness of the synchronization. To holistically evaluate this process, we quantify the DT update quality from both the accuracy and deviation perspectives.

\subsubsection{Cooperative Sensing Accuracy}
To ensure the reliability of the reconstructed digital replicas, the system imposes a fidelity requirement on the data acquisition. The cooperative sensing accuracy $\Theta_k$ for region $k$ is modeled as the average aggregated accuracy contributed by the dispatched agents:
\begin{equation}
    \Theta_k = \frac{1}{N_k} \sum_{n=1}^N \alpha_{n,k} q_{n,k}.
\end{equation}
Depending on the specific downstream applications hosted by the DT, different regions possess heterogeneous tolerance for semantic ambiguity. Therefore, the cooperative accuracy must satisfy the region-specific threshold constraint, i.e., $\Theta_k \ge \Theta_k^{\text{th}}, \forall k \in \mathcal{K}$.

\subsubsection{Twin Deviation}
By integrating the physical phases elaborated in the preceding subsections, the end-to-end closed-loop latency for a dispatched agent $n$ to accomplish the entire update workflow is summarized as:
\begin{equation}
    T_{n,k} = t_{n,k}^\text{ms} + t_n^{\text{s}} + t_{n,k}^{\text{p}} + t_{n,k}^{\text{tr}}.
\end{equation}
Since the DT orchestrator relies on the holistic fusion of semantic features extracted by all agents assigned to region~$k$, denoted by the set $\mathcal{S}_k = \{n \mid \alpha_{n,k} = 1\}$, the regional synchronization is fundamentally constrained by a straggler effect. Consequently, the effective cooperative update delay for region $k$ is determined by the maximum latency among its collaborative agents, given by $\max_{n \in \mathcal{S}_k} \{ T_{n,k} \}$.

\begin{figure}[t]
    \centering
    \includegraphics[width=\linewidth]{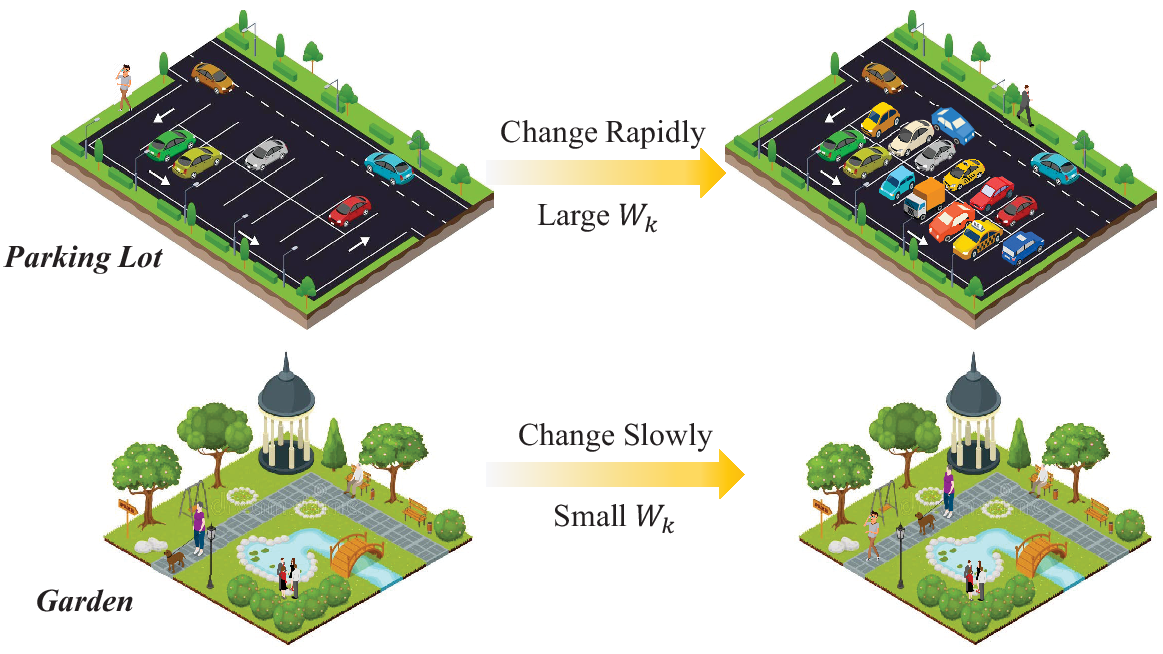}
    \caption{An example illustrating the concept of physical volatility rate $W_k$, which can be estimated from historical data by the global DT.}
    \label{fig.Wk}
\end{figure}

During the data acquisition and uploading period, the physical environment continuously evolves, causing the virtual state maintained at the BS to drift from the physical reality. Let $W_k$ denote the physical volatility rate of region $k$, which characterizes how drastically the local environment changes over time, as illustrated in Fig.~\ref{fig.Wk}. Then, we define the twin deviation $\Delta_k$ at the completion moment of the regional update as:
\begin{equation}
    \Delta_k = W_k \max_{n \in \mathcal{S}_k} \{ T_{n,k} \},
\end{equation}
which intuitively quantifies the information freshness.

\subsection{Problem Formulation}
The goal of the DT orchestrator at the BS is to intelligently dispatch the mobile agents and allocate the multidimensional communication, sensing, computation, and kinematic resources to keep the virtual replicas synchronized with the physical reality. Specifically, the system aims to minimize the maximum twin deviation across all targeted regions.

Let $\mathcal{X} \triangleq \{ \boldsymbol{\alpha}, \mathbf{v}^{\text{ms}}, \mathbf{t}^\text{s}, \boldsymbol{\rho}, \mathbf{f}, \mathbf{v}^{\text{mt}}, \mathbf{t}^\text{mt}, \mathbf{b}, \mathbf{p} \}$ denote the set of joint optimization variables. To maintain the compactness of the formulation, we define $E_{n,k}^{\text{tot}} = E_{n,k}^{\text{ms}} + E_n^{\text{s}} + E_{n,k}^{\text{c}} + E_n^{\text{mt}} + E_{n,k}^{\text{tr}}$ as the aggregate energy consumption of agent $n$ throughout the closed-loop workflow. The optimization problem is mathematically formulated as follows:
\begin{subequations}\label{eq.pf}
\begin{align}
\min_{\mathcal{X}} \quad & \max_{k \in \mathcal{K}} \ \Delta_k \tag{\theequation} \\
\text{s.t.} \hspace{1.15em}
& \frac{1}{N_k} \sum_{n=1}^N \alpha_{n,k} q_{n,k} \ge \Theta_k^{\text{th}}, \quad \forall k \in \mathcal{K}, \label{eq:c1_acc} \\
& \alpha_{n,k} \rho_n^{\min} (1 - Q_k) \le \rho_{n,k} \le 1, \quad \forall n \in \mathcal{N}, \forall k \in \mathcal{K}, \label{eq:c2_comp} \\
& \sum_{k=1}^K \alpha_{n,k} \le 1, \quad \forall n \in \mathcal{N}, \label{eq:c3_assign} \\
& \sum_{n=1}^N \alpha_{n,k} \ge 1, \quad \forall k \in \mathcal{K}, \label{eq:c4_cover} \\
& \sum_{k=1}^K \alpha_{n,k} E_{n,k}^{\text{tot}} \le E_n^{\max}, \quad \forall n \in \mathcal{N}, \label{eq:c5_energy} \\
& \sum_{n=1}^N b_n \le B_{\text{tot}}, \label{eq:c6_band} \\
& \alpha_{n,k} \in \{0, 1\}, \quad \forall n \in \mathcal{N}, \forall k \in \mathcal{K}, \label{eq:c7_binary} \\
& f_n \in (0, f_n^{\max}], \ p_n \in [0, P_n^{\max}], \ v_n^{\text{ms}}, v_n^{\text{mt}} \in (0, v_n^{\max}], \notag \\
& \ b_n \ge 0, \ t_n^{\text{s}} \ge 0, \ t_n^{\text{mt}} \ge 0, \quad \forall n \in \mathcal{N}, \label{eq:c8_bound}
\end{align}
\end{subequations}
where $E_n^{\max}$ represents the maximum available energy budget for agent $n$, $B_{\text{tot}}$ denotes the total available bandwidth, $f_n^{\max}$ and $P_n^{\max}$ designate the hardware upper bounds for the local computation capacity and the transmit power of each agent, respectively.

In problem \eqref{eq.pf}, constraint \eqref{eq:c1_acc} enforces the quality-of-service (QoS) requirement for the application-aware DT. Constraint \eqref{eq:c2_comp} bounds the DT prior-driven semantic compression ratio. Furthermore, constraints \eqref{eq:c3_assign} and \eqref{eq:c4_cover} jointly govern the topological dispatching matrix; specifically, \eqref{eq:c3_assign} prevents assignment conflicts for individual agents, while \eqref{eq:c4_cover} acts as a mandatory coverage guarantee to prevent any critical physical region from experiencing an update blackout. In terms of physical limitations, constraint \eqref{eq:c5_energy} ensures that the aggregate energy consumption does not exceed the onboard battery capacity. Finally, constraint \eqref{eq:c6_band} encapsulates the conservation of the overall spectrum resources, and \eqref{eq:c7_binary}--\eqref{eq:c8_bound} define the legitimate operational domains for the binary assignment, computation, radio frequency, kinematic, and temporal variables.

Problem \eqref{eq.pf} is a highly intractable MINLP problem. Its complexity stems from the non-smooth min-max objective, the deep coupling of binary dispatching indicators with continuous resources, and the severe non-convexity caused by fractional and logarithmic terms. This computational prohibitiveness necessitates a low-complexity joint optimization algorithm.

\section{Algorithm Design}\label{Sec:ad}
To tackle problem \eqref{eq.pf}, we propose a hierarchical two-layer optimization framework. In the outer-layer, the combinatorial multi-agent assignment is modeled as a dynamic matching game, whereas in the inner-layer, the coupled continuous resources are iteratively optimized utilizing a block coordinate descent (BCD) approach.

\subsection{Two-Layer Problem Decomposition}
To address the computational intractability of problem \eqref{eq.pf}, we decompose it into a hierarchical two-layer optimization framework. We partition the variables into the discrete assignment matrix $\boldsymbol{\alpha} = \{\alpha_{n,k}\}$ and the continuous resource set $\mathcal{X}_{\text{con}} = \{ \mathbf{v}^{\text{ms}}, \mathbf{t}^\text{s}, \boldsymbol{\rho}, \mathbf{f}, \mathbf{v}^{\text{mt}}, \mathbf{t}^\text{mt}, \mathbf{b}, \mathbf{p} \}$. Unlike standard alternating optimization, we adopt a hierarchical perspective where the outer-layer manages the combinatorial topology while the inner-layer evaluates the system performance for a given assignment.

\subsubsection{Outer-Layer Topology Matching}
The outer-layer seeks the optimal assignment matrix $\boldsymbol{\alpha}$ to minimize the system-wide synchronization bottleneck. For a specific assignment $\boldsymbol{\alpha}$, the system cost is defined as the minimum maximum twin deviation achievable under the optimized resource allocation. The outer-layer problem is formulated as:
\begin{subequations}\label{eq.outer}
\begin{align}
\min_{\boldsymbol{\alpha}} \quad & J(\boldsymbol{\alpha}) \tag{\theequation} \\
\text{s.t.} \hspace{1.15em} & \eqref{eq:c3_assign}, \eqref{eq:c4_cover}, \eqref{eq:c7_binary},
\end{align}
\end{subequations}
where $J(\boldsymbol{\alpha})$ is the value function representing the optimized objective value of the inner-layer problem for a fixed $\boldsymbol{\alpha}$.

\subsubsection{Inner-Layer Resource Allocation}
For any candidate assignment $\boldsymbol{\alpha}$ proposed by the outer-layer, the inner-layer jointly optimizes the sensing, communication, computation, and control resources to minimize the synchronization latency. This subproblem is defined as:
\begin{subequations}\label{eq.inner}
\begin{align}
\min_{\mathcal{X}_{\text{con}}} \quad & \max_{k \in \mathcal{K}} \ \Delta_k \tag{\theequation} \\
\text{s.t.} \hspace{1.15em} & \eqref{eq:c1_acc}, \eqref{eq:c2_comp}, \eqref{eq:c5_energy}, \eqref{eq:c6_band}, \eqref{eq:c8_bound}.
\end{align}
\end{subequations}

\subsection{Dynamic Matching Game with Externalities}
Based on the hierarchical decomposition, the outer-layer aims to find the optimal assignment strategy over a vast combinatorial space. We formulate this topological dispatching process as a dynamic matching game. To capture the potential idle state of agents, we introduce a virtual void region denoted by $\emptyset$, expanding the target set to $\mathcal{K}' = \mathcal{K} \cup \{\emptyset\}$. A matching state is defined as a mapping $\mu: \mathcal{N} \rightarrow \mathcal{K}'$, where $\mu(n) = k \in \mathcal{K}$ implies $\alpha_{n,k} = 1$, and $\mu(n) = \varnothing$ implies the agent is temporarily idle.

Unlike traditional bipartite matching, our framework inherently exhibits strong peer effects. According to \eqref{eq.cs}, the required sensing time $t_n^{\text{s}}$ is coupled with the total number of collaborative agents $N_k(\mu)$ in the same region. Consequently, the minimum closed-loop delay $T_{n,k}^*(\mu)$ evaluated by the inner-layer resource allocation is dependent on the strategies of other agents, defining this as a matching game with externalities. 

For notational consistency, let $\boldsymbol{\alpha}(\mu)$ denote the equivalent binary assignment matrix constructed from a specific matching state $\mu$. The system cost under topology $\mu$, denoted as $J(\mu)$, is mathematically equivalent to the outer-layer objective $J(\boldsymbol{\alpha})$ defined in \eqref{eq.outer}. Based on the inner-layer evaluation, $J(\mu)$ can be explicitly expressed as:
\begin{equation}\label{eq:cost_func}
    J(\mu) \triangleq J\left(\boldsymbol{\alpha}(\mu)\right) = \max_{k \in \mathcal{K}} \left\{ W_k \max_{n \in \mathcal{S}_k(\mu)} T_{n,k}^*(\mu) \right\},
\end{equation}
where $\mathcal{S}_k(\mu) = \{n \mid \mu(n) = k\}$ denotes the active collaborative subset for region $k$. Note that the void region $\emptyset$ does not contribute to the system cost.

To address this matching problem, the execution of the proposed algorithm proceeds in two main phases. First, we establish a fully deployed initial state to maximize the initial cooperative gain and define the operational rules bounded by a dual-stage feasibility pruning mechanism. Subsequently, we introduce the core algorithm driven by a heuristic-guided state exploration policy to efficiently refine the matching topology.

\subsubsection{Two-Phase Initialization}
To establish a high-quality initial matching state $\mu^{(0)}$ that fully exploits the available multi-agent resources, we adopt a two-phase graph-theoretic approach.

\textit{Phase I:} To strictly satisfy the fundamental coverage constraint \eqref{eq:c4_cover}, we construct a weighted bipartite graph $\mathcal{G} = (\mathcal{N}, \mathcal{K}, \mathcal{E})$. An edge $e_{n,k} \in \mathcal{E}$ is associated with a cost weight $w_{n,k} = d_{n,k}^{\text{ms}}/W_k$. We find a matching subset $\mathcal{M}_{\text{base}}^{(0)} \subset \mathcal{E}$ covering all $K$ regions while minimizing the aggregate initial cost $\sum_{(n,k) \in \mathcal{M}_{\text{base}}^{(0)}} w_{n,k}$. This is optimally solved in $\mathcal{O}(N^3)$ time utilizing the classic Hungarian algorithm.

\textit{Phase II:} Instead of leaving the remaining $N-K$ agents idle, we actively deploy them to maximize the initial cooperative sensing gain. Let $\mathcal{N}_{\text{base}}$ denote the agents assigned in phase I. For each remaining agent $n \in \mathcal{N} \setminus \mathcal{N}_{\text{base}}$, we assign it to $\mu^{(0)}(n) = \arg\min_{k \in \mathcal{K}} d_{n,k}^{\text{ms}}/W_k$. This full-deployment strategy provides a robust and highly cooperative starting point $\mu^{(0)}$ for the subsequent game.

\subsubsection{Operations and Feasibility Pruning}
To enable autonomous strategy exploration, we define two fundamental operations:
\begin{itemize}
    \item \textit{Unilateral Transfer:} An agent $n$ unilaterally changes its assignment from $\mu(n)$ to $k'$. The resulting state is denoted by $\mu_{n \to k'}$.
    \item \textit{Bilateral Swap:} Agent $n$ assigned to $\mu(n)$ and agent $n'$ assigned to $\mu(n')$ mutually exchange their tasks, yielding state $\mu^{n\leftrightarrow n'}$.
\end{itemize}

Evaluating the exact cost $J(\mu')$ for every tentative state requires triggering the inner-layer BCD optimization, which is computationally prohibitive. Therefore, we introduce a dual-stage feasibility check to reduce the invalid search space:
\begin{itemize}
    \item \textit{Pruning (Outer-Layer):} Before executing the BCD algorithm, the orchestrator evaluates the theoretical energy lower bound $E_{\text{req}}^{\min}$, which is defined as the minimum kinematic energy $E_{n,k}^{\text{ms}}$ to reach the target region. If $E_{\text{req}}^{\min} > a E_n^{\max}$ where $a\in (0,1)$ is an energy safety factor, the tentative operation is immediately rejected without invoking the inner loop.
    \item \textit{Penalty (Inner-Layer):} If the pruning is passed, the inner-layer BCD is invoked utilizing a warm-start strategy. If the BCD solver finds the continuous constraints strictly infeasible for the given topology, it returns a substantial penalty $J(\mu') = +\infty$, ensuring the operation is discarded.
\end{itemize}
An operation is approved if and only if it passes the dual-stage check, maintains basic coverage, and strictly decreases the system bottleneck, i.e., $J(\mu') < J(\mu)$.

\subsubsection{Heuristic-Guided State Exploration Policy}
To prevent aimless random walks in the vast combinatorial state space, we design a deterministic heuristic policy. At each iteration, the orchestrator generates the optimal candidate state $\mu'$ from the current state $\mu$ following a structured three-step procedure.

\textit{Step 1: Bottleneck and Donor Identification.}
The orchestrator first isolates the system bottleneck region $k^*$ and the straggler agent $n^*$ causing the maximum delay. Simultaneously, to exploit the slack resources, it identifies the safest donor region $k_{\text{donor}}$ exhibiting the minimum twin deviation:
\begin{align}
    k^* = \arg\max_{k \in \mathcal{K}} \Delta_k(\mu), \quad & n^* = \arg\max_{n \in \mathcal{S}_{k^*}(\mu)} T_{n,k^*}^*(\mu), \\
    k_{\text{donor}} =& \arg\min_{k \in \mathcal{K}} \Delta_k(\mu).
\end{align}

\textit{Step 2: Proximity-Aware Candidate Selection.}
To resolve the bottleneck efficiently while maximizing the probability of passing the energy pruning \eqref{eq:c5_energy}, we dynamically determine the optimal candidate agents outside the bottleneck region. We define the best active swap candidate $n_{\text{swap}}$ and the best donor candidate $n_{\text{donor}}$ based on their spatial proximity to $k^*$:
\begin{align}
    n_{\text{swap}} &= \arg\min_{n \in \mathcal{N} \setminus \left(\mathcal{S}_{k^*}(\mu) \cup \mathcal{S}_{\emptyset}(\mu)\right)} d_{n,k^*}^{\text{ms}}, \label{eq:cand_swap} \\
    n_{\text{donor}} &= \arg\min_{n \in \mathcal{S}_{k_{\text{donor}}}(\mu)} d_{n,k^*}^{\text{ms}}. \label{eq:cand_donor}
\end{align}
Additionally, if the void region is not empty, the best idle candidate is defined as $n_{\text{add}} = \arg\min_{n \in \mathcal{S}_{\emptyset}(\mu)} d_{n,k^*}^{\text{ms}}$.

\textit{Step 3: Best-Response Proposal Generation.}
The orchestrator constructs a restricted, high-potential candidate subset $\Omega(\mu)$ incorporating targeted operations specifically designed to dismantle the bottleneck. These operations include dropping the straggler, swapping the straggler, and executing unilateral transfers from the void or donor regions to the bottleneck:
\begin{align}\label{eq:omega_set}
    \Omega(\mu) = & \left\{ \mu_{n^* \to \emptyset} \mid |\mathcal{S}_{k^*}(\mu)| > 1 \right\} \cup \left\{ \mu^{n^*\leftrightarrow n_{\text{swap}}} \right\} \notag \\
    &\cup \left\{ \mu_{n_{\text{add}} \to k^*} \mid \mathcal{S}_{\emptyset}(\mu) \neq \emptyset \right\} \notag \\
    &\cup \left\{ \mu_{n_{\text{donor}} \to k^*} \mid |\mathcal{S}_{k_{\text{donor}}}(\mu)| > 1 \right\}.
\end{align}

To guarantee rapid and strictly monotonic convergence, the orchestrator executes a best-response exploitation over $\Omega(\mu)$. Specifically, it evaluates all feasible candidate states in $\Omega(\mu)$ that pass the outer-layer heuristic energy pruning. It then invokes the inner-layer BCD optimization for these surviving candidates and deterministically proposes the state $\mu'$ that yields the minimum system cost:
\begin{equation}\label{eq:steepest_descent}
    \mu' = \arg\min_{\hat{\mu} \in \Omega(\mu)} J(\hat{\mu}).
\end{equation}
By exhaustively evaluating the explicitly constructed high-potential subset, this best-response mechanism maximizes the marginal performance gain at each iteration.

\subsection{BCD for Resource Allocation}
For a given matching topology $\mu$, to eliminate the non-smoothness of the objective function in \eqref{eq.inner}, we introduce an auxiliary variable $\tau$ representing the maximum twin deviation among all regions. Similarly, to handle the parallel execution phase, we utilize continuous variables $t_{n,k}^{\text{p}}$ to equivalently replace the $\max$ operator in \eqref{eq.tp}. Thus, problem \eqref{eq.inner} can be equivalently reformulated as:
\begin{subequations}\label{eq.p2}
\begin{align}
\min_{\mathcal{X}_{\text{con}}, \tau, \mathbf{t}^{\text{p}}} \quad & \tau \tag{\theequation} \\
\text{s.t.} \hspace{1.6em}
& W_k \left( t_{n,k}^\text{ms} + t_n^{\text{s}} + t_{n,k}^{\text{p}} + t_{n,k}^{\text{tr}} \right) \le \tau, \notag \\
& \hspace{1.25in}  \forall k \in \mathcal{K}, \forall n \in \mathcal{S}_k(\mu), \label{eq.p2_tau} \\
& t_{n,k}^{\text{p}} \ge t_{n,k}^\text{c}, \ t_{n,k}^{\text{p}} \ge t_n^{\text{mt}}, \quad \forall n \in \mathcal{N}, \label{eq.p2_tp} \\
& \eqref{eq:c1_acc}, \eqref{eq:c2_comp}, \eqref{eq:c5_energy}, \eqref{eq:c6_band}, \eqref{eq:c8_bound}. \label{eq.p2_orig}
\end{align}
\end{subequations}
Although reformulated, problem \eqref{eq.p2} remains highly non-convex due to the deep coupling of fractional, logarithmic, and exponential terms. To tackle this intractability, we propose a low-complexity BCD algorithm that systematically decouples the continuous variables into sensing and computation, communication, and control blocks.

\subsubsection{Closed-Form Solution for $\mathbf{v}^{\text{ms}}$}
Before diving into the multi-variable block optimization, we first observe that the move-to-sense speed $\mathbf{v}^{\text{ms}}$ exhibits a monotonic decoupling property, allowing us to derive an exact analytical solution prior to the iterative numerical process.

\begin{theorem}\label{theo:vms}
Given the topological assignment $\mu(n) = k$ and the temporary resources allocated for the subsequent phases, the optimal cruising speed $(v_n^{\text{ms}})^*$ for agent $n$ during the move-to-sense phase is given by:
\begin{equation}\label{eq:vms_closed_form}
    (v_n^{\text{ms}})^* = 
    \begin{cases}
        \min \left\{ v_n^{\max}, \frac{E_n^{\text{rms}}}{d_{n,k}^{\text{ms}} \lambda_2} \right\}, & \text{if } d_{n,k}^{\text{ms}} > 0, \\
        0, & \text{if } d_{n,k}^{\text{ms}} = 0,
    \end{cases}
\end{equation}
where $E_n^{\text{rms}} = E_n^{\max} - E_n^{\text{other}} - d_{n,k}^{\text{ms}}\lambda_1$ represents the residual energy available for scaling the kinematic speed, with $E_n^{\text{other}} = E_n^{\text{s}} + E_{n,k}^{\text{c}} + E_n^{\text{mt}} + E_{n,k}^{\text{tr}}$ denoting the energy consumed in all other workflow phases.
\end{theorem}

\begin{proof}
\emph{
For any dispatched agent with a positive moving distance $d_{n,k}^{\text{ms}} > 0$, the optimization subproblem with respect to $v_n^{\text{ms}}$ aims to minimize the mobility delay $t_{n,k}^{\text{ms}} = d_{n,k}^{\text{ms}} / v_n^{\text{ms}}$. This minimization is subject to the mechanical speed limit $v_n^{\text{ms}} \le v_n^{\max}$ and the residual energy constraint $d_{n,k}^{\text{ms}}\lambda_2 v_n^{\text{ms}} \le E_n^{\text{rms}}$. Let $\nu_1 \ge 0$ and $\nu_2 \ge 0$ denote the dual variables associated with these two upper-bound constraints, respectively. The Lagrangian is formulated as:
\begin{equation}
    \mathcal{L} = \frac{d_{n,k}^{\text{ms}}}{v_n^{\text{ms}}} + \nu_1 \left( d_{n,k}^{\text{ms}}\lambda_2 v_n^{\text{ms}} - E_n^{\text{rms}} \right) + \nu_2 \left( v_n^{\text{ms}} - v_n^{\max} \right).
\end{equation}
According to the Karush-Kuhn-Tucker (KKT) stationarity condition, the optimal solution must satisfy:
\begin{equation}
    \frac{\partial \mathcal{L}}{\partial v_n^{\text{ms}}} = -\frac{d_{n,k}^{\text{ms}}}{(v_n^{\text{ms}})^2} + \nu_1 d_{n,k}^{\text{ms}} \lambda_2 + \nu_2 = 0.
\end{equation}
Since $d_{n,k}^{\text{ms}} / (v_n^{\text{ms}})^2 > 0$, the stationarity equality necessitates that $\nu_1 d_{n,k}^{\text{ms}} \lambda_2 + \nu_2 > 0$. This implies that at least one of the dual variables must be strictly positive. Based on the KKT complementary slackness condition, this guarantees that at least one of the corresponding upper-bound constraints is strictly active. Physically, since the objective function strictly monotonically decreases with $v_n^{\text{ms}}$, the optimal speed is achieved precisely at the intersection of the boundaries of the feasible domain, which directly yields \eqref{eq:vms_closed_form}.
}
\hfill $\blacksquare$
\end{proof}


\subsubsection{Optimal Sensing and Computation Block}
In the first block, we optimize $\{\mathbf{t}^\text{s}, \boldsymbol{\rho}, \mathbf{f}\}$ while fixing the communication and autonomous mobility variables. We introduce the computation workload $\mathbf{L} = \{L_n\}$ and an auxiliary non-negative variable vector $\mathbf{z} = \{z_n\}$ with $z_n = -\ln \rho_{n,k}$. By leveraging this variable substitution, we eliminate the semantic compression ratio $\boldsymbol{\rho}$ and the computation capacity $\mathbf{f}$ from the independent optimization space. Let $\mathcal{X}_{\text{sc}} \triangleq \{ \mathbf{t}^\text{s}, \mathbf{t}^{\text{c}}, \mathbf{L}, \mathbf{z} \}$ denote the continuous variable set of this block.

The primary intractability of this subproblem stems from the coupled average sensing accuracy, the highly non-convex computation energy, and the bilinear coupled delay terms. We systematically decouple these challenges as follows.

The QoS constraint \eqref{eq:c1_acc} requires the average accuracy of all collaborative agents in region $k$ to exceed $\Theta_k^{\text{th}}$. However, the sigmoid function $q_{n,k}(t_n^{\text{s}})$ in \eqref{eq.cs} is neither globally convex nor concave. To strictly address this coupled summation constraint, we employ the descent lemma to construct a guaranteed concave quadratic lower bound. At the $i$-th BCD iteration, around the local point $t_n^{\text{s},(i)}$, we bound $q_{n,k}(t_n^{\text{s}})$ as:
\begin{align}\label{eq:sca_acc}
    q_{n,k}(t_n^{\text{s}}) \ge q_{n,k}\left(t_n^{\text{s},(i)}\right) & + \nabla q_{n,k}^{(i)} \left( t_n^{\text{s}} - t_n^{\text{s},(i)} \right) \notag \\
    & - \frac{M_n}{2} \left( t_n^{\text{s}} - t_n^{\text{s},(i)} \right)^2,
\end{align}
where $\nabla q_{n,k}^{(i)}$ is the first-order derivative with respect to $t_n^{\text{s}}$, and the penalty parameter $M_n > 0$ is selected to be strictly greater than the Lipschitz constant of the gradient. By replacing the original accuracy with this surrogate function, the cooperative sensing constraint is rigorously convexified as:
\begin{equation}\label{eq:sub1_acc}
    \frac{1}{N_k} \sum_{n \in \mathcal{S}_k(\mu)} \tilde{q}_{n,k}^{(i)}(t_n^{\text{s}}) \ge \Theta_k^{\text{th}}, \quad \forall k \in \mathcal{K},
\end{equation}
where $\tilde{q}_{n,k}^{(i)}(t_n^{\text{s}})$ denotes the right-hand side of \eqref{eq:sca_acc}.

According to \eqref{eq.tc}, the required computation capacity is a dependent variable $f_n = L_n / t_{n,k}^{\text{c}}$. By substituting $f_n$, the computation energy originally formulated in \eqref{eq.ec} can be rewritten as a perspective function:
\begin{equation}
    E_{n,k}^{\text{c}} = \kappa \frac{L_n^3}{(t_{n,k}^{\text{c}})^2}.
\end{equation}
Remarkably, the fractional function $f(x,y) = x^3/y^2$ is strictly jointly convex for $x \ge 0, y > 0$. By optimizing $L_n$ and $t_{n,k}^{\text{c}}$ directly, we bypass the cubic non-convexity induced by the frequency term. Furthermore, the hardware bound $f_n \le f_n^{\max}$ is linearly recast as $L_n \le f_n^{\max} t_{n,k}^{\text{c}}$.

The required computation workload must satisfy $L_n \ge \gamma_n \eta_n t_n^\text{s} z_n$. The right-hand side is a bilinear term $t_n^\text{s} z_n$. Utilizing the arithmetic-geometric mean inequality, we construct a tight convex upper bound at the $i$-th iteration:
\begin{equation}\label{eq:sca_workload}
    t_n^\text{s} z_n \le \frac{1}{2} \left( \frac{(t_n^\text{s})^2}{\chi_n^{(i)}} + \chi_n^{(i)} z_n^2 \right) \triangleq \Upsilon_n^{(i)}(t_n^\text{s}, z_n),
\end{equation}
where the fixed coefficient is defined as $\chi_n^{(i)} = t_n^{\text{s},(i)} / z_n^{(i)}$.

Similarly, substituting $\rho_{n,k} = e^{-z_n}$, the transmission delay $t_{n,k}^{\text{tr}} = \varpi_{n,k} t_n^\text{s} e^{-z_n}$ contains a non-convex coupling, where the inverse rate $\varpi_{n,k} = \gamma_n / R_{n,k}$ acts as a predetermined constant in this block. We bound the transmission delay as:
\begin{equation}\label{eq:sca_trans}
    t_{n,k}^{\text{tr}} \le \frac{\varpi_{n,k}}{2} \left( \frac{(t_n^\text{s})^2}{\varphi_n^{(i)}} + \varphi_n^{(i)} e^{-2z_n} \right) \triangleq \tilde{t}_{n,k}^{\text{tr}},
\end{equation}
where $\varphi_n^{(i)} = t_n^{\text{s},(i)} / e^{-z_n^{(i)}}$. Since the exponential function $e^{-2z_n}$ is strictly convex, $\tilde{t}_{n,k}^{\text{tr}}$ constitutes a valid convex upper bound. Consequently, the associated transmission energy is convexified as $\tilde{E}_{n,k}^{\text{tr}} = p_n \tilde{t}_{n,k}^{\text{tr}}$.

By assembling the aforementioned convexified components, the optimization subproblem for the sensing and computation block at the $i$-th iteration is formulated as a standard convex optimization problem:
\begin{subequations}\label{eq.sub1}
\begin{align}
\min_{\mathcal{X}_{\text{sc}}, \tau, \mathbf{t}^{\text{p}}} \quad & \tau \tag{\theequation} \\
\text{s.t.} \hspace{1.5em}
& W_k \left( t_{n,k}^\text{ms} + t_n^{\text{s}} + t_{n,k}^{\text{p}} + \tilde{t}_{n,k}^{\text{tr}} \right) \le \tau, \notag \\
& \hspace{1.3in} \forall k \in \mathcal{K}, \forall n \in \mathcal{S}_k(\mu), \label{eq:sub1_tau} \\
& L_n \ge \gamma_n \eta_n \Upsilon_n^{(i)}(t_n^\text{s}, z_n), \quad \forall n \in \mathcal{S}_k(\mu), \label{eq:sub1_Ln} \\
& L_n \le f_n^{\max} t_{n,k}^{\text{c}}, \quad \forall n \in \mathcal{S}_k(\mu), \label{eq:sub1_fmax} \\
& E_{n,k}^{\text{ms}} + P_n^\text{s} t_n^\text{s} + \kappa \frac{L_n^3}{(t_{n,k}^{\text{c}})^2} + E_n^{\text{mt}} + \tilde{E}_{n,k}^{\text{tr}} \le E_n^{\max}, \notag \\
& \hspace{1.8in} \forall n \in \mathcal{S}_k(\mu), \label{eq:sub1_energy} \\
& 0 \le z_n \le -\ln \left( \rho_n^{\min} (1 - Q_k) \right), \quad \forall n \in \mathcal{S}_k(\mu), \label{eq:sub1_zn} \\
& \eqref{eq:sub1_acc}, \eqref{eq.p2_tp}, \eqref{eq:c8_bound}.
\end{align}
\end{subequations}
Problem \eqref{eq.sub1} can be efficiently solved using off-the-shelf solvers such as CVX. Once convergence is achieved, the physical variables are exactly recovered by $\rho_{n,k}^* = e^{-z_n^*}$ and $f_n^* = L_n^* / (t_{n,k}^{\text{c}})^*$. The optimal solutions obtained will update the reference points $\{ t_n^{\text{s},(i)}, z_n^{(i)} \}$ for the subsequent iterations.

\subsubsection{Optimal Communication Block}
In the second block, with the sensing, computation, and autonomous mobility variables fixed from the previous block, we jointly optimize $\{\mathbf{b},\mathbf{p}\}$. We introduce the auxiliary achievable rate $\mathbf{R} = \{R_{n,k}\}$ as an optimization variable, and let $\mathcal{X}_{\text{comm}} \triangleq \{ \mathbf{b}, \mathbf{p}, \mathbf{R} \}$ denote the continuous variable set of this block.

Mathematically, $R_{n,k}$ is the perspective function of the concave function $l(x) = \log_2\left(1 + \frac{h_{n,k}}{N_0} x\right)$, i.e., $R_{n,k}(b_n, p_n) = b_n l(p_n/b_n)$. According to the properties of perspective functions, $R_{n,k}$ is jointly concave with respect to $b_n$ and $p_n$ for $b_n > 0$ and $p_n > 0$. Therefore, the rate requirement can be equivalently represented as a convex constraint:
\begin{equation}\label{eq:sub2_rate}
    R_{n,k} - b_n \log_2 \left( 1 + \frac{p_n h_{n,k}}{N_0 b_n} \right) \le 0, \quad \forall n \in \mathcal{S}_k(\mu).
\end{equation}

Let $D_n^\text{eff} = \rho_{n,k} \gamma_n t_n^\text{s}$ denote the effective semantic payload and $T_{n}^{\text{non}} = t_{n,k}^\text{ms} + t_n^{\text{s}} + t_{n,k}^{\text{p}}$ denote the non-transmission delay. The delay constraint \eqref{eq.p2_tau} can be rewritten as:
\begin{equation}\label{eq:sub2_delay}
    W_k \left( T_{n}^{\text{non}} + \frac{D_n^\text{eff}}{R_{n,k}} \right) \le \tau, \quad \forall k \in \mathcal{K}, \forall n \in \mathcal{S}_k(\mu).
\end{equation}
Since the function $1/R_{n,k}$ is convex for $R_{n,k} > 0$, \eqref{eq:sub2_delay} defines a convex feasible region with respect to $R_{n,k}$ and $\tau$.

Furthermore, the transmission energy $E_{n,k}^{\text{tr}} = p_n t_{n,k}^{\text{tr}} = p_n (D_n^\text{eff} / R_{n,k})$ must satisfy the battery budget. By rearranging the terms, the energy constraint \eqref{eq:c5_energy} is transformed into a linear inequality:
\begin{equation}\label{eq:sub2_energy}
    p_n \le \frac{E_n^{\text{rt}}}{D_n^\text{eff}} R_{n,k}, \quad \forall n \in \mathcal{S}_k(\mu),
\end{equation}
where $E_n^{\text{rt}} = E_n^{\max} - (E_{n,k}^{\text{ms}} + E_n^{\text{s}} + E_{n,k}^{\text{c}} + E_n^{\text{mt}})$ is the residual energy.

By assembling these components, the communication block subproblem can be formulated as a standard convex program:
\begin{subequations}\label{eq.sub2}
\begin{align}
\min_{\mathcal{X}_{\text{comm}}, \tau} \quad & \tau \tag{\theequation} \\
\text{s.t.} \hspace{1.6em}
& 0 < p_n \le P_n^{\max}, \ b_n > 0, \ R_{n,k} \ge 0, \quad \forall n \in \mathcal{S}_k(\mu). \label{eq:sub2_c5} \\
& \eqref{eq:sub2_rate},\eqref{eq:sub2_delay},\eqref{eq:sub2_energy},\eqref{eq:c6_band}.
\end{align}
\end{subequations}
Problem \eqref{eq.sub2} is a convex problem that can be solved using standard optimization solvers.

\subsubsection{Optimal Control Block}
In the third block, with the sensing, computation, and communication resources determined in the preceding steps, we optimize the autonomous mobility strategy $\{\mathbf{v}^{\text{mt}}, \mathbf{t}^\text{mt}\}$. We perform a variable substitution by utilizing the exploration distance $\mathbf{d}^{\text{mt}} = \{d_n^{\text{mt}}\}$ as the direct optimization variable. Consequently, the cruising speed is expressed as $v_n^{\text{mt}} = d_n^{\text{mt}} / t_n^{\text{mt}}$. Let $\mathcal{X}_{\text{ctrl}} \triangleq \{ \mathbf{d}^{\text{mt}}, \mathbf{t}^\text{mt} \}$ denote the variable set of this block. 

The primary challenge lies in the intricate coupling between the spatial exploration distance and the location-dependent channel gain. We first establish the convexity of the resulting transmission delay.

\begin{lemma}\label{lemma:delay_conv}
Given fixed bandwidth $b_n$ and transmit power $p_n$, the transmission delay $t_{n,k}^{\text{tr}}$ is a strictly convex function with respect to the exploration distance $d_n^{\text{mt}}$.
\end{lemma}

\begin{proof}
\emph{
Let $F_n = \rho_{n,k} \gamma_n t_n^{\text{s}} / b_n$ and $A_n = p_n / (N_0 b_n)$ be constants in this block. Then, the transmission delay can be expressed as:
\begin{equation}
    t_{n,k}^{\text{tr}}(d_n^{\text{mt}}) = \frac{F_n}{\log_2\left(1 + A_n h_{n,k}(d_n^{\text{mt}})\right)},
\end{equation}
where $h_{n,k}(d_n^{\text{mt}}) = h_k^{\text{ref}} + \omega_k \left(1 - e^{-d_n^{\text{mt}}/\zeta_k}\right)$. 
First, consider the term $u(d_n^{\text{mt}}) = 1 - e^{-d_n^{\text{mt}}/\zeta_k}$. Since its second-order derivative $u''(d_n^{\text{mt}}) = -e^{-d_n^{\text{mt}}/\zeta_k}/\zeta_k^2 < 0$, $u(d_n^{\text{mt}})$ is strictly concave. Since $h_{n,k}$ is a linear mapping of $u(d_n^{\text{mt}})$, the channel gain $h_{n,k}(d_n^{\text{mt}})$ is also concave. 
Next, let $y(h) = \log_2(1 + A_n h)$. Since $y(h)$ is concave and monotonically increasing for $h > 0$, the composite function $R(d_n^{\text{mt}}) = y\left(h_{n,k}(d_n^{\text{mt}})\right)$ is strictly concave according to the composition rules. 
Finally, let $\psi(R) = F_n / R$. Since $\psi(R)$ is strictly convex and monotonically decreasing for $R > 0$, and $R(d_n^{\text{mt}})$ is strictly concave, the final composite function $t_{n,k}^{\text{tr}}(d_n^{\text{mt}}) = \psi(R(d_n^{\text{mt}}))$ is strictly convex with respect to $d_n^{\text{mt}}$.
}
\hfill $\blacksquare$
\end{proof}

By substituting $v_n^{\text{mt}} = d_n^{\text{mt}} / t_n^{\text{mt}}$ into the kinematic model, the exploration energy $E_n^{\text{mt}}$ is reformulated as $E_n^{\text{mt}} = \lambda_1 d_n^{\text{mt}} + \lambda_2 (d_n^{\text{mt}})^2 / t_n^{\text{mt}}$. Note that the term $(d_n^{\text{mt}})^2 / t_n^{\text{mt}}$ is a quadratic-over-linear perspective function, which is jointly convex for $d_n^{\text{mt}} \ge 0$ and $t_n^{\text{mt}} > 0$. Leveraging these properties, the subproblem for the control block can be formulated as a standard convex program:
\begin{subequations}\label{eq.sub3}
\begin{align}
\min_{\mathcal{X}_{\text{ctrl}}, \tau, \mathbf{t}^{\text{p}}} \quad & \tau \tag{\theequation} \\
\text{s.t.} \hspace{1.7em}
& W_k \left( t_{n,k}^\text{ms} + t_n^{\text{s}} + t_{n,k}^{\text{p}} + t_{n,k}^{\text{tr}}(d_n^{\text{mt}}) \right) \le \tau, \notag \\
& \hspace{1.3in} \forall k \in \mathcal{K}, \forall n \in \mathcal{S}_k(\mu), \label{eq:sub3_tau} \\
& t_{n,k}^{\text{p}} \ge t_{n,k}^\text{c}, \ t_{n,k}^{\text{p}} \ge t_n^{\text{mt}}, \quad \forall n \in \mathcal{S}_k(\mu), \label{eq:sub3_tp} \\
& d_n^{\text{mt}} \le v_n^{\max} t_n^{\text{mt}}, \quad \forall n \in \mathcal{S}_k(\mu), \label{eq:sub3_vmax} \\
& \lambda_1 d_n^{\text{mt}} + \lambda_2 \frac{(d_n^{\text{mt}})^2}{t_n^{\text{mt}}} + p_n t_{n,k}^{\text{tr}}(d_n^{\text{mt}}) \le E_n^{\text{rc}}, \notag \\
& \hspace{1.8in} \forall n \in \mathcal{S}_k(\mu), \label{eq:sub3_energy} \\
& d_n^{\text{mt}} \ge 0, \ t_n^{\text{mt}} \ge 0, \quad \forall n \in \mathcal{S}_k(\mu),
\end{align}
\end{subequations}
where $E_n^{\text{rc}} = E_n^{\max} - (E_{n,k}^{\text{ms}} + E_n^{\text{s}} + E_{n,k}^{\text{c}})$ represents the energy budget consumed by the fixed sensing and computation phases. Constraint \eqref{eq:sub3_vmax} equivalently maintains the mechanical speed limit. Since all constraints are convex and the objective is linear, problem \eqref{eq.sub3} can be efficiently solved using interior-point solvers. Upon obtaining the optimal $(d_n^{\text{mt}})^*$ and $(t_n^{\text{mt}})^*$, the physical cruising speed is recovered as $v_n^{\text{mt}} = (d_n^{\text{mt}})^* / (t_n^{\text{mt}})^*$.

\subsection{Algorithm Analysis}
The detailed procedures of the proposed hierarchical algorithm are summarized in Algorithm \ref{alg:overall}.

\begin{algorithm}[t]
\caption{Hierarchical Topology Dispatching and Resource Allocation Algorithm for DT Synchronization}
\label{alg:overall}
\begin{algorithmic}[1]
\renewcommand{\algorithmicrequire}{\textbf{Input:}}
\renewcommand{\algorithmicensure}{\textbf{Output:}}
\REQUIRE System parameters; tolerance $\epsilon$.
\STATE \textbf{Initialization:} 
\STATE Generate the initial full-deployment topology $\mu^{(0)}$ via Hungarian-based phase I and greedy-based phase II.
\STATE Initialize the continuous resource set $\mathcal{X}_{\text{con}}^{(0)}$, evaluate the initial system cost $J(\mu^{(0)})$, and set iteration index $j=0$.
\REPEAT
    \STATE Identify the bottleneck region $k^*$, the straggler agent $n^*$, and the donor region $k_{\text{donor}}$ to construct the high-potential candidate set $\Omega(\mu^{(j)})$ via \eqref{eq:omega_set}.
    \STATE Initialize the temporary best state: $\mu_{\text{best}} = \mu^{(j)}$ and $J_{\text{best}} = J(\mu^{(j)})$.
    \FOR{each candidate state $\hat{\mu} \in \Omega(\mu^{(j)})$}
        \STATE \textbf{if} $E_{\text{req}}^{\min} > a E_n^{\max}$ \textbf{then continue};
        \STATE Initialize BCD iteration index $i = 0$ and reference points $\{ t_n^{\text{s},(0)}, z_n^{(0)} \}$.
        \REPEAT
            \STATE Update cruising speed $(v_n^{\text{ms}})^*$ via the closed-form solution in Theorem \ref{theo:vms}.
            \STATE Fix other variables, update $\{ \mathbf{t}^\text{s}, \mathbf{t}^{\text{c}}, \mathbf{L}, \mathbf{z} \}^{(i+1)}$ by solving convex problem \eqref{eq.sub1}.
            \STATE Fix other variables, update $\{ \mathbf{b}, \mathbf{p}, \mathbf{R} \}^{(i+1)}$ by solving the convex problem \eqref{eq.sub2}.
            \STATE Fix other variables, update $\{ \mathbf{d}^{\text{mt}}, \mathbf{t}^\text{mt} \}^{(i+1)}$ by solving the convex problem \eqref{eq.sub3}.
            \STATE Update the SCA local reference points: $\{ t_n^{\text{s},(i+1)}, z_n^{(i+1)} \} \leftarrow \{ \mathbf{t}^\text{s}, \mathbf{z} \}^{(i+1)}$.
            \STATE $i \leftarrow i + 1$.
        \UNTIL{Convergence of $\tau$ within tolerance $\epsilon$.}
        \STATE \textbf{if} $\tau < J_{\text{best}}$ \textbf{then} update $J_{\text{best}} = \tau$ and $\mu_{\text{best}} = \hat{\mu}$;
    \ENDFOR
    \STATE Update the topology state $\mu^{(j+1)} \leftarrow \mu_{\text{best}}$ and $j \leftarrow j+1$.
\UNTIL{$J(\mu^{(j)}) = J(\mu^{(j-1)})$.}
\ENSURE Optimal binary assignment matrix $\boldsymbol{\alpha}^*$ and the multidimensional continuous resource set $\mathcal{X}_{\text{con}}^*$.
\end{algorithmic}
\end{algorithm}

\subsubsection{Convergence Analysis}\label{sec.ca}
The convergence of Algorithm \ref{alg:overall} is established by analyzing the inner-layer resource allocation and the outer-layer topology matching separately.

For any given topology $\mu^{(j)}$, let $\tau^{(i)}$ denote the objective value at the $i$-th inner-layer iteration. In the sensing and computation block, the successive convex approximation (SCA)-based subproblem \eqref{eq.sub1} utilizes global concave lower bounds and convex upper bounds that are tight at the local reference point. According to the properties of the block successive upper-bound minimization (BSUM) framework, solving \eqref{eq.sub1} ensures $\tau\left(\mathcal{X}_{\text{sc}}^{(i+1)}, \mathcal{X}_{\text{comm}}^{(i)}, \mathcal{X}_{\text{ctrl}}^{(i)}\right) \le \tau^{(i)}$. Since the communication and control subproblems \eqref{eq.sub2} and \eqref{eq.sub3} are solved exactly via convex optimization, we obtain a non-increasing sequence:
\begin{equation}
    \tau^{(i+1)} \le \tau\left(\mathcal{X}_{\text{sc}}^{(i+1)}, \mathcal{X}_{\text{comm}}^{(i+1)}, \mathcal{X}_{\text{ctrl}}^{(i+1)}\right) \le \tau^{(i)}.
\end{equation}
Given that the objective $\tau$ is lower-bounded by physical hardware limits, the sequence $\{\tau^{(i)}\}$ is guaranteed to converge to a stationary point of the inner-layer problem \eqref{eq.inner}.

The outer-layer manages the discrete assignment mapping $\mu: \mathcal{N} \rightarrow \mathcal{K}'$. The total number of possible matching states is finite, specifically $|\mathcal{M}| = (K+1)^N$. According to the best-response update rule in Algorithm \ref{alg:overall}, a topological operation is accepted if and only if it strictly reduces the system cost, i.e., $J(\mu^{(j+1)}) < J(\mu^{(j)})$. This strict descent property ensures that no matching state is revisited, preventing any cyclic behavior. Since the state space is finite and the sequence of system costs is strictly decreasing, the outer-layer algorithm is guaranteed to converge to a locally stable matching within a finite number of iterations.

\subsubsection{Complexity Analysis}
The computational complexity of Algorithm \ref{alg:overall} is fundamentally governed by three stages: initialization, inner-layer optimization, and outer-layer matching.

The bipartite matching and greedy reinforcement take $\mathcal{O}(N^3)$ and $\mathcal{O}((N-K)K)$ operations, respectively. Thus, the overall initialization complexity is bounded by $\mathcal{O}(N^3)$. 

For a given topology, optimizing the continuous resources involves solving the convex subproblems \eqref{eq.sub1}, \eqref{eq.sub2}, and \eqref{eq.sub3} via interior-point methods, which incurs a worst-case polynomial complexity of $\mathcal{O}(N^{3.5})$. Considering the closed-form kinematic solution \eqref{eq:vms_closed_form} adds negligible overhead, the evaluation complexity per candidate state is $\mathcal{O}(I_{\text{bcd}} N^{3.5})$, where $I_{\text{bcd}}$ is the number of BCD iterations.

At each iteration, the heuristic state generation \eqref{eq:omega_set} and the outer-layer energy pruning constrain the number of active BCD evaluations to a small constant $\tilde{C}_{\Omega}$. Assuming $I_{\text{out}}$ outer-layer iterations, the cumulative exploration complexity is $\mathcal{O}(I_{\text{out}} \tilde{C}_{\Omega} I_{\text{bcd}} N^{3.5})$.

Synthesizing these stages, the total computational complexity of the proposed framework is evaluated as $\mathcal{O} \left( N^3 + I_{\text{out}} \tilde{C}_{\Omega} I_{\text{bcd}} N^{3.5} \right)$.
Compared to the exponentially explosive exhaustive search $\mathcal{O}\left((K+1)^N\right)$ required by the original MINLP formulation, the proposed algorithm strategically reduces the computational overhead to a tractable polynomial time.

\section{Simulation Results and Analysis}\label{Sec:sra}
In this section, we present comprehensive simulation results to evaluate the performance of the proposed hierarchical algorithm. We consider a MEAN scenario in which $12$ mobile agents are randomly deployed within a $200 \times 200~\text{m}^2$ square area to cooperatively monitor $4$ target regions, with the BS located at the area center. For each scenario, the region weights $W_k$ and the sensing accuracy thresholds $\Theta_k^{\text{th}}$ are randomly drawn from $[0.5, 2.0]$ and $[0.65, 0.78]$, respectively. Unless otherwise specified, the key system parameters adopted in the simulations are summarized in Table~\ref{tab:params}, and all reported results are averaged over five independent random scenarios.

\begin{table}[t]
\caption{Main System Parameters}
\label{tab:params}
\centering
\renewcommand{\arraystretch}{1.15}
\begin{tabular}{lc}
\toprule
\textbf{Parameter} & \textbf{Value} \\
\midrule
Number of agents $N$ & $12$ \\
Number of regions $K$ & $4$ \\
Total bandwidth $B_{\text{tot}}$ & \SI{10}{MHz} \cite{you2016energy}  \\
Maximum velocity $v^{\max}_n$ & \SI{4}{m/s} \\
Maximum transmit power $P^{\max}_n$ & \SI{30}{dBm}  \\
Computation complexity $\eta$ & \SI{50}{cycles/bit} \cite{you2016energy}  \\
Path loss exponent $\delta$ & $4.0$ \cite{3GPP_TR_38901} \\
Channel reference gain $\beta_0$ & $10^{-5}$ \\
Cooperation exponent $\theta$ & $0.9$ \\
Per-agent energy budget $E_n^{\max}$ & \SI{220}{J}  \\
\bottomrule
\end{tabular}
\end{table}

\subsection{Inner-Layer Performance Analysis}

\begin{figure*}[t]
    \centering
    \vspace{-2em}
    \subfigure{
        \begin{minipage}{0.32\textwidth}
            \centering
            \includegraphics[width=1\textwidth]{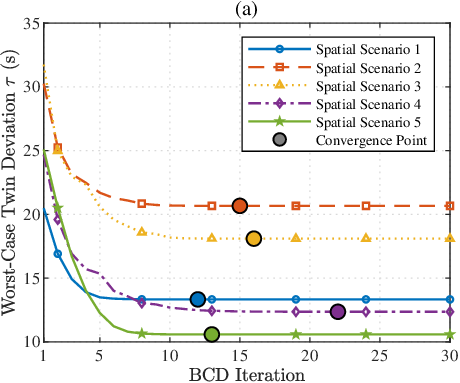}
            \label{fig.convergence}
    \end{minipage}}
    \hspace{-2mm}
    \subfigure{
        \begin{minipage}{0.32\textwidth}
            \centering
            \includegraphics[width=1\textwidth]{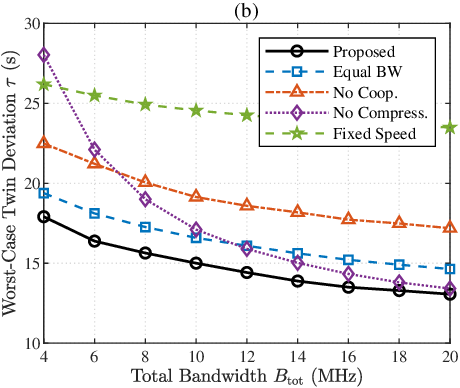}
            \label{fig.btot}	
    \end{minipage}}
        \hspace{-2mm}
        \subfigure{
            \begin{minipage}{0.32\textwidth}
                \centering
                \includegraphics[width=1\textwidth]{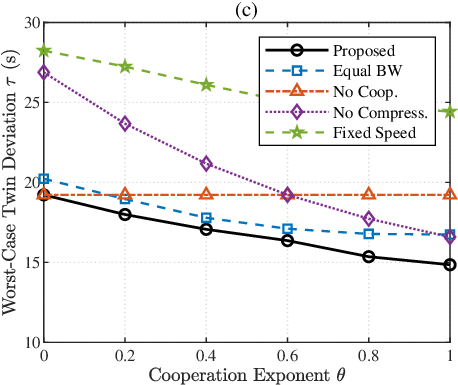}
                \label{fig.theta}	
        \end{minipage}}
        \\
        \vspace{-1.4em}
        \hspace{0.1em}
    \subfigure{
        \begin{minipage}{0.32\textwidth}
            \centering
            \includegraphics[width=1\textwidth]{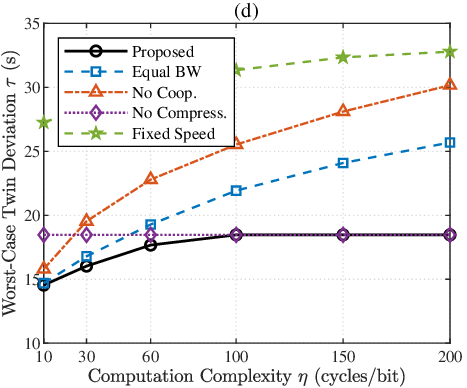}
            \label{fig.eta}
    \end{minipage}}
    \hspace{-2mm}
    \subfigure{
        \begin{minipage}{0.32\textwidth}
            \centering
            \includegraphics[width=1\textwidth]{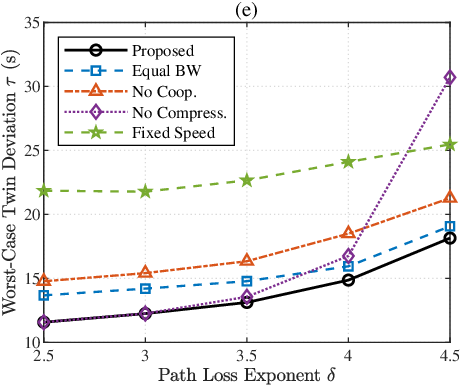}
            \label{fig.delta}	
    \end{minipage}}
        \hspace{-2mm}
        \subfigure{
            \begin{minipage}{0.32\textwidth}
                \centering
                \includegraphics[width=1\textwidth]{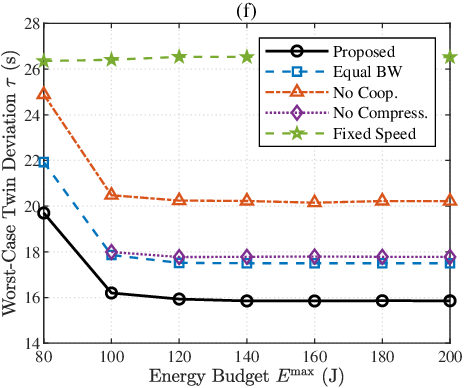}
                \label{fig.emax}	
        \end{minipage}}
    \vspace{-1.4em}
    \caption{Worst-case twin deviation versus: (a) BCD iteration, (b) total bandwidth $B_{\text{tot}}$, (c) cooperation exponent $\theta$, (d) computation complexity $\eta$, (e) path loss exponent $\delta$, (f) energy budget per agent $E^{\max}$.}
    \label{fig.BCD}
    \vspace{-1em}
\end{figure*}

We compare the proposed scheme against the following four inner-layer baseline schemes, all of which share the same outer matching solution $\boldsymbol{\alpha}^*$:
\begin{itemize}
\item ``\textbf{Equal BW}'': This scheme allocates bandwidth equally to each agent as $b_n = B_{\text{tot}}/N$.
\item ``\textbf{No Coop.}'': This scheme disables the cooperation gain modeling by setting $g(N_k) = 1$ regardless of the number of agents assigned to each region.
\item ``\textbf{No Compress.}'': This scheme fixes the semantic compression ratio at $\rho_n = 1$.
\item ``\textbf{Fixed Speed}'': This scheme holds both the move-to-sense and move-to-transmit velocities at $v^{\max}/2$.
\end{itemize}

Fig.~\ref{fig.convergence} illustrates the convergence behavior of the proposed inner BCD algorithm under five distinct initial spatial deployment topologies. All five trajectories monotonically decrease and converge within approximately $15$ iterations under the tolerance $\epsilon = 10^{-3}$. This empirical observation aligns with the theoretical analysis in Section~\ref{sec.ca}.

Fig.~\ref{fig.btot} through Fig.~\ref{fig.emax} present a comprehensive sensitivity analysis of the worst-case twin deviation $\tau$ with respect to five key system parameters, covering the communication, sensing, computation, and energy aspects of the considered MEAN. Across all five sweeps, the proposed scheme consistently attains the lowest $\tau$. Beyond this overall trend, several regime-specific phenomena merit attention. In Fig.~\ref{fig.btot} and Fig.~\ref{fig.delta}, the gap between the proposed scheme and ``No Compress.'' shrinks to less than half a second when bandwidth is abundant or the channel is favorable, but widens to over \SI{12}{s} in the opposite regime. This reveals that compression and channel resources function as substitutable means for transmission delay reduction, and semantic compression becomes critical precisely when the channel resource is constrained. In Fig.~\ref{fig.theta}, the ``No Coop.'' baseline remains constant at \SI{19.22}{s} as expected, since this baseline is structurally insensitive to the cooperation exponent. In Fig.~\ref{fig.eta}, both the proposed scheme and ``No Compress.'' converge to a common plateau of \SI{18.47}{s} for $\eta \geq 100$~cycles/bit, indicating that at high computation complexity the optimal semantic compression ratio of the proposed scheme itself degenerates to unity, since the transmission savings from compression no longer offset the computation delay penalty. Finally, in Fig.~\ref{fig.emax}, two qualitatively distinct regimes emerge. In the binding regime $E_n^{\max} \in [80, 100]$ \si{J} the energy constraint actively limits the BCD optimization, while in the saturation regime $E_n^{\max} \geq 120$~\si{J} all schemes asymptote toward their unconstrained optima. Notably, the ``Fixed Speed'' baseline saturates at approximately \SI{26.50}{s} with little response to the increasing energy budget, whereas other baselines saturate at substantially lower values around $16$ to $20$ seconds. This asymmetric saturation reveals that velocity adaptation is the primary degree of freedom that the BCD optimizer employs to navigate the energy-time trade-off; once velocity is removed from the optimization variables, the system is unable to fully exploit the relaxation of the energy budget. The proposed scheme also remains feasible across all trials throughout the $E_n^{\max}$ sweep, while ``No Compress.'' becomes infeasible in one of the five trials at $E_n^{\max} = 80$ \si{J} and is therefore omitted at that point.

\subsection{Outer-Layer Performance Analysis}
To assess the contribution of the outer matching design, we compare the proposed scheme against the following three outer-layer baseline schemes, all of which replace the outer matching procedure but retain the same inner BCD as the proposed scheme:
\begin{itemize}
\item ``\textbf{Proposed w/o TPI}'': This scheme replaces the two-phase initialization with a uniform random initialization while retaining the same best-response refinement, isolating the contribution of the intelligent initialization.
\item ``\textbf{Distance-Based}'': This scheme solves an integer linear program that minimizes the total move-to-sense distance subject to the coverage constraint.
\item ``\textbf{Random}'': This scheme assigns each agent independently and uniformly to one of $K+1$ choices with a coverage repair step, serving as an unstructured reference.
\end{itemize}

\begin{figure}
    \centering
    \includegraphics[width=0.8\linewidth]{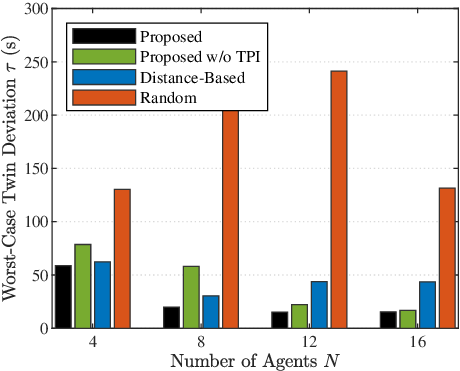}
    \caption{Worst-case twin deviation versus number of agents across diverse matching schemes.}
    \label{fig.N}
    \vspace{-1em}
\end{figure}

Fig.~\ref{fig.N} reports the worst-case twin deviation $\tau$ versus the number of agents $N$. As $N$ grows from $4$ to $16$, the proposed scheme exhibits a monotonically decreasing trend that saturates for $N \geq 12$. This phenomenon can be attributed to the diminishing marginal benefit of additional agents, given the limitation of the total bandwidth. The gap of roughly \SI{28}{s} between the proposed scheme and ``Distance-Based'' for $N \geq 12$ quantifies the cost of ignoring the physical volatility rate $W_k$ and the cooperation gain, since these two factors are absent from the pure mobility-distance objective adopted by ``Distance-Based''. A more revealing observation concerns the ``Proposed w/o TPI'' curve, where the gap to the proposed scheme shrinks from more than \SI{30}{s} at $N = 8$ to less than \SI{2}{s} at $N = 16$. This trend suggests that the two-phase initialization is most valuable in resource-constrained regimes where the search landscape is rugged, whereas in agent-abundant regimes the best-response refinement alone suffices to drive a random initialization toward a near-optimal solution. The ``Random'' curve exhibits high variance and a non-monotonic mean across the sweep, owing to the large inherent variance of uniform random sampling.

Fig.~\ref{fig.visualization} presents a side-by-side visualization of the matching topologies produced by the four schemes on a common scenario at $N = 12$ and $K = 4$, with the resulting $J^*$ value annotated in each panel. In this scenario, the physical volatility rates of the four regions are $W_1 = 1.83$, $W_2 = 1.68$, $W_3 = 1.78$, and $W_4 = 1.44$. In Fig.~\ref{fig.proposed}, the proposed scheme attains $J^*$ of \SI{35.93}{s} with the assignment distribution $N_k = [1, 2, 2, 6]$ and one void agent. Despite holding the lowest physical volatility rate among the four regions, $R_4$ receives six cooperating agents. This phenomenon can be attributed to the overall geometrical distribution of the agents and regions. 
Fig.~\ref{fig.no_tpi}, ``Proposed w/o TPI'', reaches $J^*$ of \SI{57.21}{s} with a sparser deployment $N_k = [3, 1, 2, 2]$ and four void agents, illustrating that without the two-phase initialization the best-response refinement converges to a locally stable but under-deployed configuration. Fig.~\ref{fig.distance}, ``Distance-Based'', achieves $J^*$ of \SI{93.50}{s} with $N_k = [2, 2, 2, 6]$ and zero voids. This baseline attains a worse $J^*$ than the proposed scheme since ``Distance-Based'' exhaustively deploys all agents, whereas the proposed scheme strategically idles one agent to free its bandwidth and energy budget for redistribution among the active agents. Finally, Fig.~\ref{fig.random}, ``Random'', produces $J^*$ of \SI{123.24}{s} with a visibly imbalanced workload $N_k = [5, 1, 2, 3]$.

\begin{figure}[t]
    \centering
    \subfigure{
        \begin{minipage}{0.236\textwidth}
            \centering
            \includegraphics[width=1\textwidth]{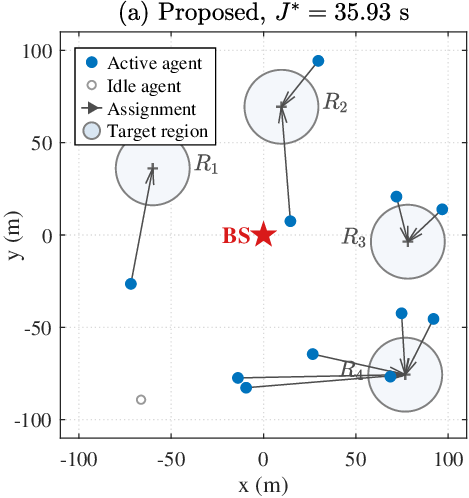}
            \label{fig.proposed}
    \end{minipage}}
    \hspace{-2mm}
        \subfigure{
            \begin{minipage}{0.236\textwidth}
                \centering
                \includegraphics[width=1\textwidth]{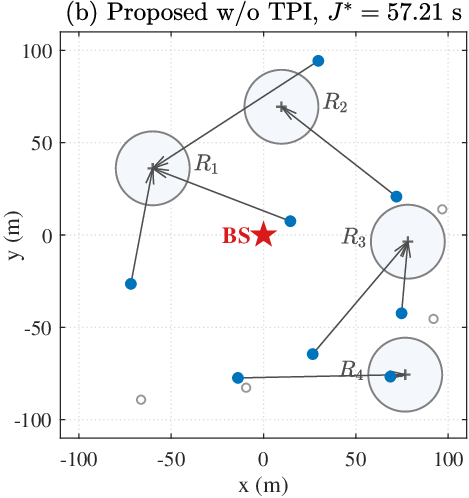}
                \label{fig.no_tpi}	
        \end{minipage}}
        \\
        \vspace{-2em}
    \subfigure{
        \begin{minipage}{0.236\textwidth}
            \centering
            \includegraphics[width=1\textwidth]{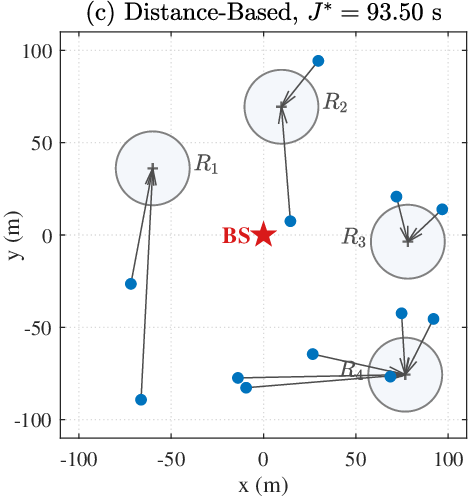}
            \label{fig.distance}
    \end{minipage}}
        \hspace{-2mm}
        \subfigure{
            \begin{minipage}{0.236\textwidth}
                \centering
                \includegraphics[width=1\textwidth]{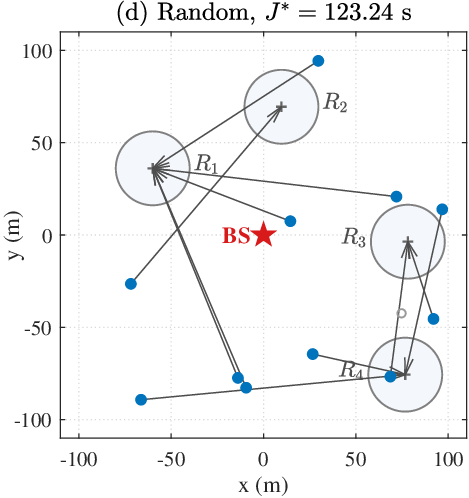}
                \label{fig.random}
        \end{minipage}}
    \vspace{-1.4em}
    \caption{Visualization of the matching topologies across different schemes.}
    \label{fig.visualization}
    \vspace{-1em}
\end{figure}



\section{Conclusion}\label{Sec:c}
In this paper, we have established a comprehensive framework for DT synchronization over agentic AI-empowered MEANs. By investigating the five-stage closed-loop workflow, we transitioned from a static sensing-communication paradigm to a dynamic sensing-moving-computing-transmission synergy. Our analyses reveal two key insights for future agentic networks: semantic compression effectively compensates for constrained wireless bandwidth in highly volatile environments, while autonomous velocity adaptation serves as an essential degree of freedom to navigate the fundamental energy-time trade-off.

Future research will focus on integrating channel knowledge maps to enhance the precision of mobility-gain modeling and extending the coordination to multi-cell scenarios.

\bibliographystyle{IEEEtran}
\bibliography{ref}

@article{ZHAO2024107055,
title = {A joint communication and computation design for semantic wireless communication with probability graph},
journal = {J. Franklin Inst.},
volume = {361},
number = {13},
pages = {107055},
year = {2024},
month={Sep.},
issn = {0016-0032},
doi = {https://doi.org/10.1016/j.jfranklin.2024.107055},
author = {Zhouxiang Zhao and Zhaohui Yang and Xu Gan and Quoc-Viet Pham and Chongwen Huang and Wei Xu and Zhaoyang Zhang},
}

@ARTICLE{10550151,
  author={Zhao, Zhouxiang and Yang, Zhaohui and Huang, Chongwen and Wei, Li and Yang, Qianqian and Zhong, Caijun and Xu, Wei and Zhang, Zhaoyang},
  journal={IEEE Internet Things J.}, 
  title={A Joint Communication and Computation Design for Distributed {RISs} Assisted Probabilistic Semantic Communication in {IIoT}}, 
  year={2024},
  month={Aug.},
  volume={11},
  number={16},
  pages={26568-26579},
  doi={10.1109/JIOT.2024.3409271}}

@ARTICLE{10915662,
  author={Zhao, Zhouxiang and Yang, Zhaohui and Hu, Ye and Zhu, Chen and Shikh-Bahaei, Mohammad and Xu, Wei and Zhang, Zhaoyang and Huang, Kaibin},
  journal={IEEE Trans. Commun.}, 
  title={Compression Ratio Allocation for Probabilistic Semantic Communication With {RSMA}}, 
  year={2025},
  month={Sep.},
  volume={73},
  number={9},
  pages={7304-7318},
  doi={10.1109/TCOMM.2025.3548689}}

@ARTICLE{11006980,
  author={Zhao, Zhouxiang and Yang, Zhaohui and Chen, Mingzhe and Zhu, Chen and Xu, Wei and Zhang, Zhaoyang and Huang, Kaibin},
  journal={IEEE Trans. Wireless Commun.}, 
  title={Energy-Efficient Probabilistic Semantic Communication Over Space-Air-Ground Integrated Networks}, 
  year={2025},
  month={Oct.},
  volume={24},
  number={10},
  pages={8814-8829},
  doi={10.1109/TWC.2025.3569102}}

@ARTICLE{11208653,
  author={Yang, Yinchao and Ding, Yahao and Yang, Zhaohui and Huang, Chongwen and Zhang, Zhaoyang and Niyato, Dusit and Shikh-Bahaei, Mohammad R.},
  journal={IEEE Internet Things J.}, 
  title={Toward Efficient and Privacy-Aware {eHealth} Systems: An Integrated Sensing, Computing, and Semantic Communication Approach}, 
  year={2025},
  month={Dec.},
  volume={12},
  number={24},
  pages={55087-55105},
  doi={10.1109/JIOT.2025.3623776}}

@ARTICLE{1638342,
  author={Yongguo Mei and Yung-Hsiang Lu and Hu, Y.C. and Lee, C.S.G.},
  journal={IEEE Trans. Robot.}, 
  title={Deployment of mobile robots with energy and timing constraints}, 
  year={2006},
  month={Jun.},
  volume={22},
  number={3},
  pages={507-522},
  doi={10.1109/TRO.2006.875494}}

@ARTICLE{10841377,
  author={Wang, Jiaxiang and Yang, Yinchao and Yang, Zhaohui and Huang, Chongwen and Chen, Mingzhe and Zhang, Zhaoyang and Shikh-Bahaei, Mohammad},
  journal={IEEE Trans. Cogn. Commun. Netw.}, 
  title={Generative {AI} Empowered Semantic Feature Multiple Access {(SFMA)} Over Wireless Networks}, 
  year={2025},
  month={Apr.},
  volume={11},
  number={2},
  pages={791-804},
  doi={10.1109/TCCN.2025.3529692}}

@ARTICLE{11370176,
  author={Jiang, Feibo and Pan, Cunhua and Wang, Kezhi and Michiardi, Pietro and Dobre, Octavia A. and Debbah, Merouane},
  journal={IEEE J. Sel. Areas Commun.}, 
  title={From Large {AI} Models to Agentic {AI}: A Tutorial on Future Intelligent Communications}, 
  year={2026},
  volume={},
  number={},
  pages={},
  doi={10.1109/JSAC.2026.3660010}}

@ARTICLE{11339915,
  author={Zhang, Ruichen and Liu, Guangyuan and Liu, Yinqiu and Zhao, Changyuan and Wang, Jiacheng and Xu, Yunting and Niyato, Dusit and Kang, Jiawen and Li, Yonghui and Mao, Shiwen and Sun, Sumei and Shen, Xuemin and Kim, Dong In},
  journal={IEEE Commun. Surveys Tuts.}, 
  title={Toward Edge General Intelligence With Agentic {AI} and Agentification: Concepts, Technologies, and Future Directions}, 
  year={2026},
  volume={28},
  number={},
  pages={4285-4318},
  doi={10.1109/COMST.2026.3651702}}

@ARTICLE{11373363,
  author={Nguyen, Thuan Minh and Truong, Vu Tuan and Le, Long Bao},
  journal={IEEE Internet Things Mag.}, 
  title={Agentic {AI} Meets Edge Computing in Autonomous {UAV} Swarms}, 
  year={2026},
  volume={},
  number={},
  pages={},
  doi={10.1109/MIOT.2025.3645767}}

@ARTICLE{11366905,
  author={Zhao, Changyuan and Liu, Guangyuan and Zhang, Ruichen and Liu, Yinqiu and Wang, Jiacheng and Kang, Jiawen and Niyato, Dusit and Li, Zan and Shen, Xuemin and Han, Zhu and Sun, Sumei and Yuen, Chau and Kim, Dong In},
  journal={IEEE Trans. Cogn. Commun. Netw.}, 
  title={Edge General Intelligence Through World Models, Large Language Models, and Agentic {AI}: Fundamentals, Solutions, and Challenges}, 
  year={2026},
  volume={12},
  number={},
  pages={5649-5675},
  doi={10.1109/TCCN.2026.3658762}}

@ARTICLE{11303308,
  author={Zhang, Ruichen and Tang, Shunpu and Liu, Yinqiu and Niyato, Dusit and Xiong, Zehui and Sun, Sumei and Mao, Shiwen and Han, Zhu},
  journal={IEEE Commun. Mag.}, 
  title={Toward Agentic {AI}: Generative Information Retrieval Inspired Intelligent Communications and Networking}, 
  year={2026},
  month={Jan.},
  volume={64},
  number={1},
  pages={197-204},
  doi={10.1109/MCOM.001.2500073}}

@ARTICLE{11373008,
  author={Du, Baoxia and Li, Ruidong and Niyato, Dusit and Su, Zhou},
  journal={IEEE Netw.}, 
  title={{TaskON}: Task-Oriented Networking for Agentic {AI}}, 
  year={2026},
  volume={},
  number={},
  pages={},
  doi={10.1109/MNET.2026.3657581}}

@ARTICLE{11297177,
  author={Liu, Guangyuan and Liu, Yinqiu and Zhang, Ruichen and Du, Hongyang and Niyato, Dusit and Xiong, Zehui and Sun, Sumei and Jamalipour, Abbas},
  journal={IEEE Commun. Mag.}, 
  title={Wireless Agentic {AI} with Retrieval-Augmented Multimodal Semantic Perception}, 
  year={2026},
  month={Jan.},
  volume={64},
  number={1},
  pages={230-236},
  doi={10.1109/MCOM.001.2500293}}

@ARTICLE{11298134,
  author={Liu, Yinqiu and Liu, Guangyuan and Wang, Jiacheng and Zhang, Ruichen and Niyato, Dusit and Sun, Geng and Xiong, Zehui and Han, Zhu},
  journal={IEEE J. Sel. Areas Commun.}, 
  title={{LAMeTA}: Intent-Aware Agentic Network Optimization via a Large {AI} Model-Empowered Two-Stage Approach}, 
  year={2025},
  volume={},
  number={},
  pages={},
  doi={10.1109/JSAC.2025.3642840}}

@ARTICLE{11207697,
  author={Li, Haiyuan and Madhukumar, Hari and Methley, Nicholas and Chen, Xuewen and Wu, Yulei and Parra-Ullauri, Juan and Sharma, Vishnu and Lee, Jeongran and Koblitz, Arndt Ryo and Andrews, Matthew and Liu, Sige and Deng, Yansha and Kolawole, Oluwatayo Y. and Tassi, Andrea and Simeonidou, Dimitra},
  journal={IEEE Internet Things J.}, 
  title={Future Factories with {6G}: Agentic {AI} and Cyber-Physical Digital Twins}, 
  year={2025},
  volume={},
  number={},
  pages={},
  doi={10.1109/JIOT.2025.3623075}}

@ARTICLE{11207716,
  author={Wang, Yuntao and Guo, Shaolong and Pan, Yanghe and Su, Zhou and Chen, Fahao and Luan, Tom H. and Li, Peng and Kang, Jiawen and Niyato, Dusit},
  journal={IEEE Trans. Cogn. Commun. Netw.}, 
  title={Internet of Agents: Fundamentals, Applications, and Challenges}, 
  year={2026},
  volume={12},
  number={},
  pages={4476-4501},
  doi={10.1109/TCCN.2025.3623369}}

@ARTICLE{11303197,
  author={Liu, Ziheng and Zhang, Jiayi and Zhu, Yiyang and Shi, Enyu and Xu, Bokai and Niyato, Dusit and Jin, Shi and Ai, Bo},
  journal={IEEE J. Sel. Areas Commun.}, 
  title={{MaLAM4Com}: Multi-Agent Cooperative Large {AI} Models for Wireless Communications}, 
  year={2025},
  volume={},
  number={},
  pages={},
  doi={10.1109/JSAC.2025.3645753}}

@ARTICLE{11022699,
  author={Wang, Yuntao and Pan, Yanghe and Su, Zhou and Deng, Yi and Zhao, Quan and Du, Linkang and Luan, Tom H. and Kang, Jiawen and Niyato, Dusit},
  journal={IEEE Commun. Surveys Tuts.}, 
  title={Large Model-Based Agents: State-of-the-Art, Cooperation Paradigms, Security and Privacy, and Future Trends}, 
  year={2026},
  volume={28},
  number={},
  pages={1906-1949},
  doi={10.1109/COMST.2025.3576176}}

@ARTICLE{9899718,
  author={Mihai, Stefan and Yaqoob, Mahnoor and Hung, Dang V. and Davis, William and Towakel, Praveer and Raza, Mohsin and Karamanoglu, Mehmet and Barn, Balbir and Shetve, Dattaprasad and Prasad, Raja V. and Venkataraman, Hrishikesh and Trestian, Ramona and Nguyen, Huan X.},
  journal={IEEE Commun. Surveys Tuts.}, 
  title={Digital Twins: A Survey on Enabling Technologies, Challenges, Trends and Future Prospects}, 
  year={2022},
  month={4th Quart.},
  volume={24},
  number={4},
  pages={2255-2291},
  doi={10.1109/COMST.2022.3208773}}

@INPROCEEDINGS{zhang2025transferring,
  author={Zhang, Zifan and Fang, Minghong and Chen, Mingzhe and Liu, Yuchen},
  booktitle={Proc. 2025 Int. Conf. Model. Anal. Simu. Wireless Mobile Syst. (MSWiM)}, 
  title={On Transferring, Merging, and Splitting Task-Oriented Network Digital Twins}, 
  year={2025},
  month={Oct.},
  volume={},
  number={},
  pages={695-704},
  doi={10.1109/MSWiM67937.2025.11309199}}

@ARTICLE{yu2025optimizing,
  author={Yu, Hanzhi and Liu, Yuchen and Yang, Zhaohui and Sun, Haijian and Chen, Mingzhe},
  journal={IEEE Internet Things J.}, 
  title={Optimizing Wireless Resource Management and Synchronization in Digital Twin Networks}, 
  year={2025},
  month={Aug.},
  volume={12},
  number={15},
  pages={29152-29163},
  doi={10.1109/JIOT.2025.3543126}}

@ARTICLE{10070572,
  author={Zheng, Jinkai and Luan, Tom H. and Zhang, Yao and Li, Rui and Hui, Yilong and Gao, Longxiang and Dong, Mianxiong},
  journal={IEEE Trans. Wireless Commun.}, 
  title={Data Synchronization in Vehicular Digital Twin Network: A Game Theoretic Approach}, 
  year={2023},
  month={Nov.},
  volume={22},
  number={11},
  pages={7635-7647},
  doi={10.1109/TWC.2023.3254158}}

@ARTICLE{li2025delay,
  author={Li, Bin and Cai, Haichen and Liu, Lei and Fei, Zesong},
  journal={IEEE Trans. Veh. Technol.}, 
  title={Delay-Aware Digital Twin Synchronization in Mobile Edge Networks With Semantic Communications}, 
  year={2025},
  month={Jul.},
  volume={74},
  number={7},
  pages={10974-10983},
  doi={10.1109/TVT.2025.3548844}}

@INPROCEEDINGS{cakir2023synchronize,
  author={Cakir, Lal Verda and Al-Shareeda, Sarah and Oktug, Sema F. and Özdem, Mehmet and Broadbent, Matthew and Canberk, Berk},
  booktitle={Proc. 2023 IEEE 28th Int. Workshop Comput. Aided Model. Design Commun. Links Netw. (CAMAD)}, 
  title={How to synchronize Digital Twins? A Communication Performance Analysis}, 
  year={2023},
  month={Nov.},
  volume={},
  number={},
  pages={123-127},
  doi={10.1109/CAMAD59638.2023.10478422}}

@ARTICLE{10530992,
  author={Tang, Jianhang and Nie, Jiangtian and Bai, Jingpan and Xu, Ji and Li, Shaobo and Zhang, Yang and Yuan, Yanli},
  journal={IEEE Internet Things J.}, 
  title={{UAV}-Assisted Digital-Twin Synchronization With Tiny-Machine-Learning-Based Semantic Communications}, 
  year={2024},
  month={Sep.},
  volume={11},
  number={17},
  pages={28437-28451},
  doi={10.1109/JIOT.2024.3401229}}

@ARTICLE{elloumi2025spectrum,
  author={Elloumi, Mohamed and Hassan, Md. Zoheb and Kaddoum, Georges},
  journal={IEEE Trans. Netw. Service Manage.}, 
  title={Spectrum Sharing in Internet-of-Vehicles Networks: Digital Twin-Empowered Proactive Interference Management Approach}, 
  year={2025},
  month={Aug.},
  volume={22},
  number={4},
  pages={3228-3248},
  doi={10.1109/TNSM.2025.3541977}}

@ARTICLE{liu2024two,
  author={Liu, Wenshuai and Fu, Yaru and Guo, Yongna and Lee Wang, Fu and Sun, Wen and Zhang, Yan},
  journal={IEEE Trans. Wireless Commun.}, 
  title={Two-Timescale Synchronization and Migration for Digital Twin Networks: A Multi-Agent Deep Reinforcement Learning Approach}, 
  year={2024},
  month={Nov.},
  volume={23},
  number={11},
  pages={17294-17309},
  doi={10.1109/TWC.2024.3452689}}

@ARTICLE{9921194,
  author={Lee, Joash and Cheng, Yanyu and Niyato, Dusit and Guan, Yong Liang and González G., David},
  journal={IEEE Trans. Veh. Technol.}, 
  title={Intelligent Resource Allocation in Joint Radar-Communication With Graph Neural Networks}, 
  year={2022},
  month={Oct.},
  volume={71},
  number={10},
  pages={11120-11135},
  doi={10.1109/TVT.2022.3187377}}

@ARTICLE{9380899,
  author={Yates, Roy D. and Sun, Yin and Brown, D. Richard and Kaul, Sanjit K. and Modiano, Eytan and Ulukus, Sennur},
  journal={IEEE J. Sel. Areas Commun.}, 
  title={Age of Information: An Introduction and Survey}, 
  year={2021},
  month={May},
  volume={39},
  number={5},
  pages={1183-1210},
  doi={10.1109/JSAC.2021.3065072}}

@INPROCEEDINGS{guo2024age,
  author={Guo, Yongna and Fu, Yaru and Zhang, Yan and Quek, Tony Q. S.},
  booktitle={Proc. GLOBECOM 2024 - 2024 IEEE Global Commun. Conf.}, 
  title={Age-of-Information and Energy Optimization in Digital Twin Edge Networks}, 
  year={2024},
  month={Dec.},
  volume={},
  number={},
  pages={3194-3200},
  doi={10.1109/GLOBECOM52923.2024.10901362}}

@INPROCEEDINGS{10000964,
  author={Shu, Yiling and Wang, Zhao and Liao, Haijun and Zhou, Zhenyu and Nasser, Nidal and Imran, Muhammad},
  booktitle={Proc. GLOBECOM 2022 - 2022 IEEE Global Commun. Conf.}, 
  title={Age-of-Information-Aware Digital Twin Assisted Resource Management for Distributed Energy Scheduling}, 
  year={2022},
  month={Dec.},
  volume={},
  number={},
  pages={5705-5710},
  doi={10.1109/GLOBECOM48099.2022.10000964}}

@ARTICLE{11183623,
  author={Chen, Xiangyi and Xing, Huanlai and Zhang, Qinnan and Xiong, Zehui and Bi, Yuanguo and Wang, Xingwei},
  journal={IEEE Trans. Netw. Sci. Eng.}, 
  title={Age-of-Information-Aware Hierarchical Collaborative Inference in Digital Twin-Enabled Vehicular Edge Computing}, 
  year={2026},
  volume={13},
  number={},
  pages={2388-2404},
  doi={10.1109/TNSE.2025.3615533}}

@ARTICLE{11004606,
  author={Noroozi, Kiana and Todd, Terence D. and Zhao, Dongmei and Karakostas, George},
  journal={IEEE Trans. Veh. Technol.}, 
  title={Age of Information in Digital Twin Migration}, 
  year={2025},
  month={Oct.},
  volume={74},
  number={10},
  pages={16281-16294},
  doi={10.1109/TVT.2025.3570505}}

@ARTICLE{liu2025joint,
  author={Liu, Zechen and Liu, Xin and Yang, Wenyi and Zhang, Xueyan},
  journal={IEEE Trans. Wireless Commun.}, 
  title={Joint Sensing and Age of Information Optimization for Energy Constrained {UAV}-Assisted Integrated Sensing, Calculation, and Communication}, 
  year={2025},
  month={May},
  volume={24},
  number={5},
  pages={4440-4453},
  doi={10.1109/TWC.2025.3539108}}

@ARTICLE{liao2024integration,
  author={Liao, Haijun and Lu, Jiaxuan and Shu, Yiling and Zhou, Zhenyu and Tariq, Muhammad and Mumtaz, Shahid},
  journal={IEEE J. Sel. Topics Signal Process.}, 
  title={Integration of {6G} Signal Processing, Communication, and Computing Based on Information Timeliness-Aware Digital Twin}, 
  year={2024},
  month={Jan.},
  volume={18},
  number={1},
  pages={98-108},
  doi={10.1109/JSTSP.2023.3341353}}

@ARTICLE{10460140,
  author={Tang, Lun and Cheng, Zhangchao and Dai, Jun and Zhang, Hongpeng and Chen, Qianbin},
  journal={IEEE Trans. Veh. Technol.}, 
  title={Joint Optimization of Vehicular Sensing and Vehicle Digital Twins Deployment for {DT}-Assisted {IoVs}}, 
  year={2024},
  month={Aug.},
  volume={73},
  number={8},
  pages={11834-11847},
  doi={10.1109/TVT.2024.3373175}}

@INPROCEEDINGS{9812038,
  author={Xu, Runsheng and Xiang, Hao and Xia, Xin and Han, Xu and Li, Jinlong and Ma, Jiaqi},
  booktitle={Proc. 2022 Int. Conf. Robot. Auto. (ICRA)}, 
  title={{OPV2V}: An Open Benchmark Dataset and Fusion Pipeline for Perception with Vehicle-to-Vehicle Communication}, 
  year={2022},
  month={May},
  volume={},
  number={},
  pages={2583-2589},
  doi={10.1109/ICRA46639.2022.9812038}}

@inproceedings{yang2023spatio,
  title={Spatio-temporal domain awareness for multi-agent collaborative perception},
  author={Yang, Kun and Yang, Dingkang and Zhang, Jingyu and Li, Mingcheng and Liu, Yang and Liu, Jing and Wang, Hanqi and Sun, Peng and Song, Liang},
  booktitle={Proc. IEEE/CVF Int. Conf. Compu. Vis. (ICCV)},
  pages={23383--23392},
  year={2023}
}

@article{Bonilla2016MDA,
  author  = {{Bonilla Licea}, Daniel and Ghogho, Mounir and McLernon, Des and Zaidi, Syed Ali Raza},
  title   = {Mobility Diversity-Assisted Wireless Communication for Mobile Robots},
  journal = {IEEE Trans. Robot.},
  volume  = {32},
  number  = {1},
  pages   = {214--229},
  year    = {2016},
  month   = feb,
  doi     = {10.1109/TRO.2015.2513745}
}

@article{Miyagusuku2018GPPathLoss,
  author  = {Renato Miyagusuku and Atsushi Yamashita and Hajime Asama},
  title   = {Precise and accurate wireless signal strength mappings using Gaussian processes and path loss models},
  journal = {Robot. Auto. Systems},
  volume  = {103},
  pages   = {134--150},
  year    = {2018},
  doi     = {10.1016/j.robot.2018.02.011}
}

@ARTICLE{11414114,
  author={Zhou, Kequan and Zhang, Guangyi and Li, Hanlei and Cai, Yunlong and Liu, Shengli and Yu, Guanding},
  journal={IEEE Trans. Commun.}, 
  title={{F4-CKM}: Learning Channel Knowledge Map With Radio Frequency Radiance Field Rendering}, 
  year={2026},
  volume={74},
  number={},
  pages={5684-5700},
  doi={10.1109/TCOMM.2026.3668162}}

@INPROCEEDINGS{9879243,
  author={Yu, Haibao and Luo, Yizhen and Shu, Mao and Huo, Yiyi and Yang, Zebang and Shi, Yifeng and Guo, Zhenglong and Li, Hanyu and Hu, Xing and Yuan, Jirui and Nie, Zaiqing},
  booktitle={Proc. IEEE/CVF Conf. Comput. Vis. Pattern Recog. (CVPR)}, 
  title={{DAIR-V2X}: A Large-Scale Dataset for Vehicle-Infrastructure Cooperative {3D} Object Detection}, 
  year={2022},
  volume={},
  number={},
  pages={21329-21338},
  doi={10.1109/CVPR52688.2022.02067}}

@inproceedings{xu2022v2x,
  title={V2x-vit: Vehicle-to-everything cooperative perception with vision transformer},
  author={Xu, Runsheng and Xiang, Hao and Tu, Zhengzhong and Xia, Xin and Yang, Ming-Hsuan and Ma, Jiaqi},
  booktitle={Proc. Eur. Conf. Comput. Vis. (ECCV)},
  pages={107--124},
  year={2022}
}

@article{zhao2026agentic,
  title={Agentic {AI}-Empowered Wireless Agent Networks With Semantic-Aware Collaboration via {ILAC}},
  author={Zhao, Zhouxiang and Wang, Jiaxiang and Yang, Zhaohui and Yang, Kun and Zhang, Zhaoyang and Chen, Mingzhe and Huang, Kaibin},
  journal={arXiv preprint arXiv:2604.02381},
  year={2026}
}

@techreport{3GPP_TR_38901,
  author       = {{3GPP}},
  title        = {Study on Channel Model for Frequencies from 0.5 to 100 {GHz}},
  institution  = {3rd Generation Partnership Project (3GPP)},
  number       = {TR 38.901, version 17.0.0},
  year         = {2022},
  month        = {Mar.},
  type         = {Tech. Rep.}
}

@article{you2016energy,
  title={Energy-efficient resource allocation for mobile-edge computation offloading},
  author={You, Changsheng and Huang, Kaibin and Chae, Hyukjin and Kim, Byoung-Hoon},
  journal={IEEE Trans. Wireless Commun.},
  volume={16},
  number={3},
  pages={1397--1411},
  year={2016},
  month={Mar.},
  publisher={IEEE}
}

@article{wu2026joint,
  title={Joint Communication and Computation Design for Mobile Embodied {AI} Network {(MEAN)}},
  author={Wu, Chenliang and Zhao, Zhouxiang and Wang, Jiaxiang and Xu, Ruopeng and Zhu, Chen and Yang, Zhaohui and Zhang, Zhaoyang},
  journal={arXiv preprint arXiv:2605.14300},
  year={2026}
}

@article{chen2026split,
  title={Split and Aggregation Learning for Foundation Models Over Mobile Embodied {AI} Network {(MEAN)}: A Comprehensive Survey},
  author={Chen, Qianzhou and Sun, Siqi and Xu, Minrui and Ji, Sijie and Kang, Jiawen and Mao, Yijie and Zhao, Zhouxiang and Yang, Zhaohui and Niyato, Dusit},
  journal={arXiv preprint arXiv:2605.00970},
  year={2026}
}

@article{zhao2025agentic,
  title={Agentic AI for Low-Altitude Semantic Wireless Networks: An Energy Efficient Design},
  author={Zhao, Zhouxiang and Yi, Ran and Cang, Yihan and Jin, Boyang and Yang, Zhaohui and Chen, Mingzhe and Huang, Chongwen and Zhang, Zhaoyang},
  journal={arXiv preprint arXiv:2509.19791},
  year={2025}
}

\end{document}